\documentclass[sigconf]{aamas} 

\usepackage{balance} 

\title[Temporal Network Creation Games: The Impact of Non-Locality and Terminals]{Temporal Network Creation Games:\\The Impact of Non-Locality and Terminals}

\author{Davide Bilò}
\affiliation{
  \institution{University of L’Aquila}
  \city{L'Aquila}
  \country{Italy}}
\email{davide.bilo@univaq.it}

\author{Sarel Cohen}
\affiliation{
  \institution{Hasso Plattner Institute}
  \city{Potsdam}
  \country{Germany}}
\email{sarel.cohen@hpi.de}

\author{Tobias Friedrich}
\affiliation{
  \institution{Hasso Plattner Institute}
  \city{Potsdam}
  \country{Germany}}
\email{tobias.friedrich@hpi.de}

\author{Hans Gawendowicz}
\affiliation{
  \institution{Hasso Plattner Institute}
  \city{Potsdam}
  \country{Germany}}
\email{hans.gawendowicz@hpi.de}

\author{Nicolas Klodt}
\affiliation{
  \institution{Hasso Plattner Institute}
  \city{Potsdam}
  \country{Germany}}
\email{nicolas.klodt@hpi.de}

\author{Pascal Lenzner}
\affiliation{
  \institution{Hasso Plattner Institute}
  \city{Potsdam}
  \country{Germany}}
\email{pascal.lenzner@hpi.de}

\author{George Skretas}
\affiliation{
  \institution{Hasso Plattner Institute}
  \city{Potsdam}
  \country{Germany}}
\email{george.skretas@hpi.de}

\begin{abstract}
We live in a world full of networks where our economy, our communication, and even our social life crucially depends on them. These networks typically emerge from the interaction of many entities, which is why researchers study agent-based models of network formation. While traditionally static networks with a fixed set of links were considered, a recent stream of works focuses on 
networks whose behavior may change over time. In particular, Bil{\`{o}} {\sl et al.} (IJCAI 2023) 
recently introduced a game-theoretic network formation model that embeds temporal aspects in networks. More precisely, a network is formed by selfish agents corresponding to nodes in a given host network with edges having labels denoting their availability over time. Each agent strategically selects local, i.e., incident, edges to ensure temporal reachability towards everyone at low cost.

In this work we set out to explore the impact of two novel conceptual features: agents are no longer restricted to creating incident edges, called the \emph{global setting}, and agents might only want to ensure that they can reach a subset of the other nodes, called the \emph{terminal model}. For both, we study the existence, structure, and quality of equilibrium networks. For the terminal model, we prove that many core properties crucially depend on the number of terminals. We also develop a novel tool that allows translating equilibrium constructions from the non-terminal model to the terminal model. For the global setting, we show the surprising result that equilibria in the global and the local model are incomparable and we establish a high lower bound on the Price of Anarchy of the global setting that matches the upper bound of the local model. This shows the counter-intuitive fact that allowing agents more flexibility in edge creation does not improve the quality of equilibrium networks. Finally, in contrast to Bil{\`{o}} {\sl et al.} (IJCAI 2023) where the authors restrict the model to single labels per connection, all of our results hold for the restricted case and the generalized case where every edge can have multiple labels.
\end{abstract}

         
\newcommand{\BibTeX}{\rm B\kern-.05em{\sc i\kern-.025em b}\kern-.08em\TeX}

\usepackage{colortbl}
\usepackage{amsthm}
\usepackage{cleveref}
\usepackage{tikz}
\usetikzlibrary{positioning}
\usetikzlibrary{bending}
\usetikzlibrary{arrows.meta}
\usepackage{pgfplots}
\usepgfplotslibrary{colormaps}
\usepackage{thmtools}
\usepackage{dsfont}
\usepackage{mathtools}
\usepackage[loadonly]{enumitem}
\usepackage{paralist}

\definecolor{myYellow}{rgb}{1 1 0.878}

\newcommand*{\N}{\mathds{N}}

\newcommand*{\R}{\mathds{R}}

\newcommand*{\reach}[1]{R_{#1}}

\newcommand*{\cost}[3]{c_{#3}(#1,#2)} 
\DeclareMathOperator{\socialcost}{SC}
\newcommand*{\SC}[2]{\socialcost_{#2}(#1)} 
\newcommand{\lifetime}{\lambda^{max}}
\newcommand{\timeEdges}{\Lambda}

\newcommand*{\product}{\Pi}

\newcommand*{\s}{\mathbf{s}}
\def\poa {{\sf PoA}}
\def\pos {{\sf PoS}}
\def\opt {{\sf OPT}}
\def\lne {{\sf NE}^{\sf lo}}
\def\lge {{\sf GE}^{\sf lo}}
\def\gne {{\sf NE}^{\sf gl}}
\def\gge {{\sf GE}^{\sf gl}}
\def\NE {{\sf NE}}
\def\GE {{\sf GE}}
\def\lo {{\sf lo}}
\def\gl {{\sf gl}}

\newcommand\nod{}
\def\nod(#1,#2,#3,#4){\node[draw, circle,  minimum size=15pt, inner sep=0.5pt](#1)[#2=#4 of #3]{$#1$};}
\newcommand\ed{}
\def\ed(#1,#2,#3){(#1) edge node[above, inner sep=2pt,]{$#3$} (#2)}
\newcommand\edc{}
\def\edc(#1,#2,#3){(#1) edge node[centered, inner sep=0.5pt,fill=white,circle]{$#3$} (#2)}
\newcommand\edo{}
\def\edo(#1,#2,#3,#4){(#1) edge node[#4, inner sep=2pt,]{$#3$} (#2)}

\tikzstyle{nod} = [inner sep=2pt, draw, circle, minimum size=15pt]

\begin{document}

\pagestyle{fancy}
\fancyhead{}

\maketitle 


\section{Introduction}
Networks are an integral part of our everyday lives, playing a key role in almost every aspect of human existence. Prominent examples include transportation networks (road networks, train tracks, airplaine routes, etc.), communication networks (e.g. the Internet), neural networks (both biological and artificial), biological networks (e.g. protein-protein interaction networks) and many more. With the growing digitization of society, networks, in particular communication networks and (online) social networks, came more and more into the focus of computer science research over the last decades. Many different topics have been studied ranging from the formation of social networks~\cite{KleinbergGroupFormation} over information diffusion~\cite{FriedrichRummorSpreading} and generating synthetic social networks with real-world properties~\cite{RandomGraphsForSocialNetworks,krioukov2010hyperbolic} to uncover their underlying geometry~\cite{papadopoulos2012popularity}.

To understand how social networks (and many other types of networks) emerge, one must understand the mechanisms and principles that govern the formation of networks among several non-cooperative agents~\cite{Papadimitriou01}. This sparked the investigation of game-theoretic network formation models like the \emph{Network Creation Game} (NCG)~\cite{fabrikant03}. In this model, selfish agents act as nodes of a network which can form costly connections to others to gain a central position in the arising network. In particular, each agent can build connections only locally, i.e., via creating incident edges. Since then, many variations and extensions of this model have been formulated and studied, e.g., variants with non-uniform edge cost~\cite{MeiromMO14,Cord-LandwehrMH14,ChauhanLMM17,BFLLM21}, robustness considerations~\cite{MeiromMO15,ChauhanLMM16,Goyal16,Echzell0LM20}, or geometric aspects~\cite{MoscibrodaSW06a,EidenbenzKZ06,bilo2019}.

Although all these models aim to capture time-dependent processes of network formation, in practice, they consider networks that, once formed, are static. This is in contrast to many real-world networks in which temporal aspects play a prominent role. We highlight two motivating examples to make this more evident.

One example is the commercial airline network: each time an airline company wants to serve a new route, the company also has to take into account connecting flights with their corresponding departure and arrival times. Planning the routes carefully can ensure reachability: customers can get from any airport to any other airport by taking a sequence of flights, possibly of different airlines, with ascending departure and arrival times. Here, the airlines are the selfish agents that can establish new connections to enable their customers to travel anywhere. 

For another example, consider the supply chain network of companies that are participating in the production of a particular product X.  Assume that company A wants to make product X and sell it.  Unless company A owns every part of the production chain (which is highly unlikely in today's world), they want to have a connection to other companies in order to send materials and use their means of production that are missing from their production chain. As such, they want to guarantee that they have the logistical infrastructure to send their parts to all other companies participating. But company A may want to combine deliveries. For example, load a vehicle with parts that goes to company B, and then the vehicle loads up parts from company B and moves them to company C. In order for this behaviour to be accurately portrayed, the scheduling of the connections must happen in ascending order (time-wise).

Other examples of network formation that include temporal properties are scheduling problems in which jobs have an order of preferences, 
neural networks where neurons 
forming a chain 
are serially activated one after the other, 
navigation networks in which the travel time of roads changes over time (e.g. due to traffic, or roadblocks),
as well as pathways in biological networks which are series of actions among molecules in a cell that lead to a certain product or a change in the cell. These examples motivate, that understanding network formation of temporal networks is crucial.

Recently, \citet{tempNCG_ijcai} made a first step towards incorporating temporal aspects into NCGs.
In their model, the game is played on an underlying \emph{temporal host network} that defines the time steps in which the bought edges will be available and each agent can only build incident edges.
However, this setting might not be general enough to represent real-world networks. Let us consider our two previous motivating examples again. 

In the airline route network, the 5th Freedom Right\footnote{\url{https://www.icao.int/Pages/freedomsAir.aspx}} allows airline companies to create connections among countries that do not necessarily include the country the airline is based at. Meanwhile, a company is not interested in reaching every possible destination in other countries, but it mainly serves the hubs and cities which are in high demand for its customers. Finally, an airline company may want to have multiple connections between two countries on each day.

Similarly, in the supply chain network, company~A will send parts to company~B for processing and then may want to use its own transport vehicles to transfer the processed parts to company~C afterwards. Additionally, company A may not need to have a connection to the whole supply chain network, but only to particular other companies. Finally, company A may want to establish more than one connection between two factories during a day, due to a multitude of logistical reasons.

In this work, we extend the model by~\citet{tempNCG_ijcai} to cope with the three raised issues. First, we introduce the \emph{terminal model} in which nodes want to reach only a subset of the nodes, called terminals. The second addition is the \emph{global setting} in which we allow each agent to build connections anywhere in the network, i.e., agents can create non-incident edges. Finally, in contrast to \citet{tempNCG_ijcai} where the authors restrict the model to single labels per connection, we study the restricted case and also generalize to multiple labels per connection.

Before giving an overview of our contribution, we introduce our model and some notation.

\subsection{Model and Notation}
We first introduce temporal graphs, then we move on to the game-theoretic definition of our model.

\paragraph{Temporal Graphs and Temporal Spanners}
A \emph{temporal graph} $G=(V_G, E_G, \lambda_G)$ consists of a set of nodes $V_G$, a set of undirected edges $E_G\subseteq\{\{u,v\}\subseteq V_G\mid u\neq v\}$, and a labeling function $\lambda_G\colon E_G\rightarrow P(\N)\setminus\emptyset$, where, for each edge $e \in E_H$, the term $\lambda_G(e)$ denotes the set of \emph{time labels} of $e$. Informally, the labeling function $\lambda_G$ describes the time steps in which edge $e$ is available. We sometimes write $\lambda_G(e)+c$ for some $c\in\N$ to denote the set $\{\lambda+c\mid \lambda\in\lambda_G(e)\}$. We define the set $\timeEdges_G$ of \emph{time edges} as the set of tuples of edges and each of their time labels, i.e. $\timeEdges_G \coloneqq \{(e, \lambda)\mid e\in E_G, \lambda \in \lambda_G(e)\}$. For nodes $u,v\in V_G$ and a time label $\lambda$, we sometimes abuse notation and write $(u,v,\lambda)$ instead of $(\{u,v\},\lambda)$. Furthermore, we call the largest label $\lifetime_G\coloneqq\max_{e\in E_G}\max_{\lambda \in \lambda_G(e)}\lambda$ the \emph{lifetime} of $G$. If the graph $G$ is clear from context, we might omit the subscript $G$ to enhance readability. We call a temporal graph simple if there is exactly one time label on each edge. For simple graphs $G$, we sometimes treat $\lambda_G(e)$ as a number instead of a set for easier notation.

A \emph{temporal path} is a sequence of nodes $v_0,\dots, v_\ell\in V$ that form a path in $G$, such that there exists an increasing sequence of time labels $\lambda_0\leq\dots\leq \lambda_{\ell-1}$ with $\lambda_i \in \lambda(\{v_i,v_{i+1}\})$ for every $i=0,\ldots,\ell-1$.
We define $\ell$ to be the length of the temporal path. Note that we do not require the labels on the temporal path to increase strictly.
We say that a node $u\in V_G$ \emph{reaches} $v\in V_G$ if there is a temporal path from $u$ to $v$ in $G$. Observe that, even though the edges are undirected, a temporal path from $u$ to $v$ does not necessarily imply the existence of a temporal path from $v$ to $u$. Moreover, we define  $R_G(v)\subseteq V_G$ as the set of nodes that node $v$ can reach in $G$. We call the graph $G$ \emph{temporally connected} if $R_G(v)=V_G$ for every node $v \in V_G$.

We define a \emph{temporal host graph with terminals} (or \emph{host graph} for short) as $H=(V_H, E_H, \lambda_H, T_H)$, where $(V_H, E_H, \lambda_H)$ is a complete temporal graph, i.e. $E_H=\{\{u,v\}\subseteq V_H\mid u\neq v\}$, while $T_H\subseteq V_H$ is a set of \emph{terminal nodes} (or terminals), which is the same for all agents. W.l.o.g., we assume that, for every $\tau=1,\ldots, \lambda_H^{\max}$, there is an edge $e \in E_H$ with $\tau \in \lambda_H(e)$.\footnote{Indeed, as long as some value of $\tau$, with $1\leq \tau \leq \lambda_H^{\max}$, is missing, we can decrease by 1 all the edge labels that are strictly larger than $\tau$.}

A temporal subgraph of $H$ is a temporal graph $G$ such that $(V_G,E_G)$ is a subgraph of $(V_H,E_H)$ and  $\lambda_G(e)\subseteq\lambda_H(e)$ for every $e \in E_G$.
A \emph{terminal spanner} of $H$ is a temporal subgraph $G$ of $H$, with $V_G=V_H$, where every node reaches all the terminals, i.e., $T_H\subseteq R_G(v)$ for every $v \in V_H$. Note that each terminal also needs to reach all the other terminals. Furthermore, for $k=n$ this is the definition of a \emph{temporal spanner}.

\paragraph{Game-Theoretic Model}
We introduce the game-theoretic model that we study in this paper. Let $H$ be a temporal host graph with terminals that serves as a host graph for our game. Each node $v\in V_H$ is a selfish agent whose \emph{strategy} $S_v\subseteq \timeEdges_H$ corresponds to the set of time edges that agent~$v$ buys. We distinguish two settings: \emph{Global edge-buying}, where agents have no restrictions on the time edges they can buy, and \emph{local edge-buying} where agents can only buy incident time edges, i.e. $S_v\subseteq \{(\{v,u\},\lambda)\mid u\in V_H\setminus\{v\}\}$. We denote by $\s=\bigcup_{v\in V_H}\{(v,S_v)\}$ the \emph{strategy profile} and by $G(\s)$ the temporal graph formed by the agents. Formally, the graph $G(\s)$ is a temporal subgraph of $H$ with $V_G=V_H$ and $\timeEdges_{G(\s)}=\bigcup_{(v,S_v)\in \s}S_v$. Note that $E_{G(\s)}$ and $\lambda_{G(\s)}$ are implicitly defined when $\timeEdges_{G(\s)}$ is known. In figures, we sometimes display edges as directed to illustrate the edge ownership. Such edges are bought by the node they originate in and can still be used in both direction for the purpose of temporal reachability. In the global setting this simplification does not always work. In this case we write onto the edge who buys it. For simple temporal graphs we sometimes talk about buying edges instead of time edges as they are equivalent in this case.

Each agent $v\in V_H$ aims at reaching all terminals while buying as few time edges as possible. Formally, agent $v$ wants to minimize its costs given by
\begin{align*}
    \cost{v}{\s}{H}&=|S_v|+C\cdot |T\setminus R_{G(\s)}(v)|.
\end{align*}
where $C>1$ is a large constant ensuring that reaching any terminal is more beneficial than not buying a single edge. Indeed, as $H$ is a complete temporal graph, each agent $v$ can always reach all terminals in $T_H$ by buying, for example, an arbitrary time edge for each edge of the form $\{v,u\}$, with $u \in T_H$.
We call the defined models \emph{global edge-buying $k$-terminal Temporal Network Creation Game} (global $k$-tNCG) and \emph{local edge-buying $k$-terminal Temporal Network Creation Game} (local $k$-tNCG), respectively.

Before defining the solution concepts, we need some more notation regarding strategies. Let $\s$ be a strategy profile and consider any agent $v\in V_H$. We define $\s_{-v}\coloneqq \s\setminus \{(v,S_v)\}$ as the strategy profile without the strategy of agent $v$. Now, consider an alternative strategy $S_v'\neq S_v$ for agent $v$. We denote by $\s_{-v}\cup S_v'$ the strategy profile $\s_{-v}\cup\{(v,S_v')\}$. If $\cost{v}{\s_{-v}\cup S_v'}{H}<\cost{v}{\s}{H}$, we say that $\s_{-v}\cup S_v'$ is an \emph{improving response} for $v$ (w.r.t. $\s$). If additionally, the strategies $S_v$ and $S_v'$ differ by at most one element (i.e. $v$ either adds or removes a single time edge), we call this a \emph{greedy improving response}\footnote{Note that, in the literature~\cite{Lenzner12}, a greedy improving response also allows a swap, i.e. removing one edge and adding one edge simultaneously. However, in our game, every improving response consisting of a swap also implies an improving response that only adds an edge and omits the remove part. This is because a swap is an improving response for an agent only when the number of reached terminals increases. This means, we can disregard swaps for our definition of greedy improving responses.}. We call $\s$ a \emph{best response} of agent $v$ (resp., a \emph{greedy best response}) if there is no improving response (resp., greedy improving response) for agent $v$.

We can now introduce our solution concepts. A strategy profile $\s$ is a \emph{Pure Nash Equilibrium (NE)} (resp., \emph{Greedy Equilibrium (GE)}) if no agent has an improving response (resp., greedy improving response).
As every greedy improving response is also an improving response, we have that every NE is also a GE. Furthermore, every NE (and thus every GE) guarantees pairwise disjoint strategies, since any agent can trivially remove the intersection of its strategy and some other agent's strategy without affecting its reachability. Moreover, our definition of the cost function directly implies that the created graph $G(\s)$ must be a terminal spanner.

Lastly, we introduce a measure for the well-being of all agents combined. Let $H$ be a host graph and let $\s$ be any strategy profile. The \emph{social cost} of $\s$ on $H$ is then defined as
\begin{align*}
    \socialcost_H(\s)=\sum_{v\in V_H}\cost{v}{\s}{H}.
\end{align*}
Note that $\socialcost_H(\s)=|\timeEdges_{G(\s)}|$ for every NE or GE $\s$. A strategy profile of minimum social cost for the given host graph $H$ is called \emph{social optimum} and denoted as $\s_H^*$. When considering the efficiency of equilibria, we will compare their social costs to the social optimum. For $n,k\in\N$ with $k\le n$, let $\mathcal{H}_{n,k}$ be the set of all host graphs containing $n$ nodes and $k$ terminals. Furthermore, for a host graph $H$, let $\lne_H$, $\gne_H$, $\lge_H$ and $\gge_H$ be the sets of Nash Equilibria and Greedy Equilibria in the local edge-buying and the global edge-buying setting, respectively. We define the \emph{Price of Anarchy (PoA)} for the local edge-buying setting with respect to Nash Equilibria as the ratio of the socially worst equilibrium and the social optimum
\begin{align*}
    \poa^\lo_\NE(n,k)\coloneqq \max_{H\in\mathcal{H}_{n,k}}\max_{\s\in \lne_H}\frac{\socialcost_H(\s)}{\socialcost_H(s^*_H)}.
\end{align*}
We define $\poa^\gl_\NE$, $\poa^\lo_\GE$, and $\poa^\gl_\GE$ analogously. If a result holds for both settings (local and global), we omit the superscript. If a result holds for both GE and NE, we omit the subscript.

Lastly, we define the \emph{Price of Stability} as
\begin{align*}
    \pos^\lo_\NE(n,k)\coloneqq \max_{H\in\mathcal{H}_{n,k}}\min_{\s\in \lne_H}\frac{\socialcost_H(\s)}{\socialcost_H(s^*_H)}.
\end{align*}
Again, we define $\pos^\gl_\NE$, $\pos^\lo_\GE$, and $\pos^\gl_\GE$ analogously.

\subsection{Our Contribution}
\begin{table*}
	\centering\setlength{\tabcolsep}{3pt}%
	\renewcommand{\arraystretch}{1.3}%
    \begin{tabular}{c|c|>{\columncolor{myYellow}}c|>{\columncolor{myYellow}}c}
            & (local $n$-)TNCG & local $k$-TNCG & global $k$-TNCG
        \\ \hline
		      Optimum &
            min temporal spanner &
            min terminal spanner &
            min terminal spanner
		\\ \hline
		      Equilibria &
            \begin{tabular}{@{}c@{}}$\lifetime=2\colon$spanning tree\\$m\leq\sqrt{6}n^\frac{3}{2}$ \end{tabular}&
            \begin{tabular}{@{}c@{}}$\lifetime=2\colon$spanning tree [\ref{cor:NE_lifetime_two}]\\ $k=2\colon$exists 
            [\ref{thm:NE_two_terminals}]\\ $m\leq\sqrt{6k}n+n$  [\ref{thm:dense_not_ge}] \end{tabular}&
            \begin{tabular}{@{}c@{}}$\lifetime=2\colon$spanning tree [\ref{cor:NE_lifetime_two}]\\$k=2\colon$exists [\ref{thm:NE_two_terminals}]\\GE: exists [\ref{cor:GEexists}]\end{tabular}
		\\ \hline
		      PoA &
            \begin{tabular}{@{}c@{}}$\mathcal{O}(\sqrt{n})$\\ $\mathcal{O}(\lifetime)$ \\ $\Omega(\log n)$ \\ $\poa_\GE\le\mathcal{O}(\log(n))\poa_\NE$\end{tabular} &
            \begin{tabular}{@{}c@{}}$\mathcal{O}(\sqrt{k})$ [\ref{thm:PoA_upper_bound_local}]\\ $\mathcal{O}(\lifetime)$ [\ref{thm:PoALifetime}]\\ $\Omega(\log k)$ [\ref{thm:PoA_lower_bound_local}]\end{tabular} &
            \begin{tabular}{@{}c@{}}$\mathcal{O}(k)$ [\ref{thm:PoA_upper_bound_global}]\\ $\mathcal{O}(\lifetime)$ [\ref{thm:PoALifetime}]\\ $\Omega(\sqrt{k})$ [\ref{thm:PoALB}]\\$\poa_\GE \in \Theta(k)$[\ref{thm:PoA_upper_bound_global},\ref{thm:PoA_lower_bound_global_GE}]\end{tabular}
		\\ \hline
		      PoS &
            ? &
            ? &
            1 (for GE) [\ref{cor:GEPoS}]
	\end{tabular}
    \caption{An overview of our results (yellow) and comparison with the existing results from~\cite{tempNCG_ijcai}. Here, $n$ is the number of nodes, $m$ the number of (time) edges, $k$ the number of terminals, and $\lifetime$ the largest label in the host graph.}\label{table:results}
\end{table*}

The main contribution of this work is the generalization of the model introduced by~\citet{tempNCG_ijcai} and its game-theoretic analysis. We introduce the concepts of terminals, global edge-buying and multiple labels.  To the best of our knowledge terminals have not been considered yet on any network creation model. While the terminal version is just a generalization of the normal model, we show that the global edge-buying leads to a completely different model with an incomparable set of equilibrium graphs.   Our results for the generalized model work for both single label graphs and multi label graphs. Note that our techniques can be used to extend the results of~\citet{tempNCG_ijcai} to the multi label model. \Cref{table:results} gives an overview of our results in comparison with the results of \cite{tempNCG_ijcai}. 

In \Cref{sec:equilibria}, we study the structure and properties of equilibria. First, we introduce a special kind of graph product, see \Cref{def:graph_product}, that allows us to take any two host networks and respective equilibria and construct a new host graph together with a new equilibrium. This can then be used to construct lower bound examples for the PoA for a wide range of numbers of nodes $n$ and numbers of terminals $k$ by constructing only a few initial equilibria. Additionally, we show that, in the local setting, many structural properties of equilibria and bounds on the price of anarchy derived by \cite{tempNCG_ijcai} that seemed to be dependent on the number of nodes in the graph are actually dependent on the number of terminals instead. Moreover, we show that for two terminals in the local and global setting, Greedy and Nash Equilibria always exist. We also show that for the global setting, Greedy Equilibrium graphs are exactly the set of inclusion minimal temporal spanners. We conclude the section by showing that the  set of equilibrium graphs in the global setting are incomparable to the ones from the local setting.

In \Cref{sec:poa}, we analyze the efficiency of equilibria. For the global setting, many results carry over from the local setting but there are notable differences. Our findings show that allowing the agents to buy non-incident edges does not improve the efficiency of equilibria but in fact might make them even worse. For the case of Greedy Equilibria, we show that the PoA in the global setting is in $\Omega(k)$, in contrast to the upper bound of $O(\sqrt{k})$ that exists for the local setting. We also show that for Nash Equilibria, the PoA is in $\mathcal{O}(\sqrt{k})$ in the local setting, while it is in $\Omega(\sqrt{k})$ for the global setting. While it is still possible that those bounds match asymptotically, we conjecture that the actual PoA is much closer to the lower bound of $\Omega(\log k)$ in the local setting.

\subsubsection*{Simple Host Graphs}
As mentioned before, all our results also hold for the special case where the host graph is a simple temporal graph, i.e. every edge has exactly one time label. For all results from \Cref{def:graph_product} to \Cref{lem:NE_dismountable} this is true since given simple host graphs, the constructions in turn admit simple host graphs. All remaining results are either general upper bounds/statements, and therefore, they also hold for the special case of simple graphs or constructions that are already simple graphs.

\subsection{Related Work}
As mentioned in the introduction, this paper extends the temporal network creation game proposed by~\citet{tempNCG_ijcai}, which studies the all-pairs reachability in the local edge-buying model. In particular, in~\cite{tempNCG_ijcai}, the authors first prove the existence of NE for host graphs with lifetime $\lambda_H^{\max}=2$ and show that, for every host graph with lifetime $\lambda_H^{\max}\geq 2$, the problems of computing a best response strategy and the problem of deciding whether a strategy profile is a NE are both NP-hard. The authors then consider upper and lower bounds to the PoA w.r.t. both NE and GE. In particular, they show that the PoA w.r.t. NE is in between $\Omega(\log n)$ and $\mathcal{O}(\sqrt{n})$. Moreover, they connect GE with NE by showing that the PoA w.r.t. GE is no more than a $\mathcal{O}(\log n)$ factor away the PoA w.r.t. NE.

Besides the paper by~\citet{tempNCG_ijcai} which, to the best of our knowledge, is the only one that combines temporal aspects with network formation games, there has been an extensive line of research on related games in the last decades. One of the earliest models which is close to our work is by~\citet{bala2000noncooperative}, where selfish agents buy incident edges and their utility increases with the number of agents they can reach while it decreases with the number of edges bought.
For the version where undirected edges are formed, the authors prove that equilibria always exist forming either stars or empty graphs, and that improving response dynamics quickly converge to such states. They also show  how to efficiently compute a best response strategy as well as deciding if a given state is in equilibrium. \citet{Goyal16} extended this model to a setting with attacks on the formed network, where the objective is to maintain post-attack reachability. This variant is more complex, yet \citet{FriedrichIKLNS17} proved that best response strategies can still be computed efficiently. Recently, \citet{ChenJKKM19} studied a variant where the attacks are probabilistic. \citet{EidenbenzKZ06} studied the different, yet related, topology control game, where the agents are points in the plane and edge costs are proportional to the Euclidean distance between the endpoints. A similar game was studied by~\citet{gulyas2015navigable}, with the difference that agents are points in hyperbolic space and using greedy routing.
Regarding the idea of using global edge-buying in network creation games, the model by \citet{DemaineHMZ09} is related. There, coalitions of agents can buy costshares of any edge in the network.

From a centralized algorithmic perspective, starting from the work by~\citet{KKK02}, a lot of research has been devoted to the problem of computing sparse spanners in temporal graphs. More precisely, temporal cliques admit sparse temporally connected spanners~\cite{CasteigtsPS21}, even when we seek for all-pairs temporal paths of bounded length~\cite{BiloDG0R22}. In contrast, there exist very dense temporal graphs that are not complete whose temporal spanners are all dense~\cite{AxiotisF16}. Closely related to the reachability problem, \citet{klobasMMS22} study the problem of finding the minimum number of labels required to achieve temporal connectivity in a graph.


\section{Equilibria}\label{sec:equilibria}
\begin{figure*}
    \centering
    \includegraphics[scale=0.9]{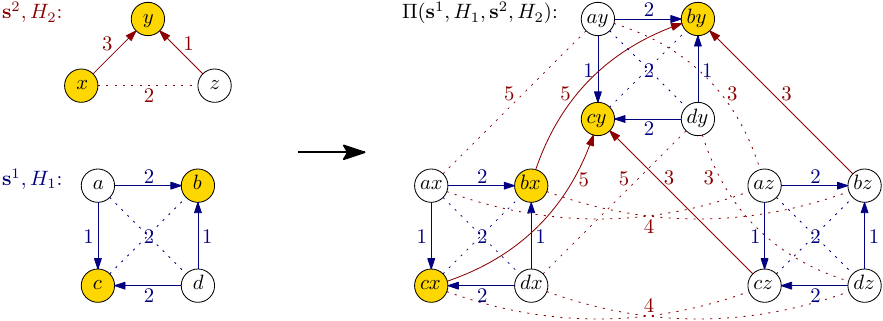}
    \caption{This figure shows two host graphs $H_1$ and $H_2$ (dotted and solid lines) and two respective strategy profiles $s^1$ and $s^2$ (solid lines) on the left. Yellow nodes are terminals and all edges are bought by the nodes where they originate. On the right, you can see the resulting graph product according to \Cref{def:graph_product}. For clarity, all edges with label 6 are not displayed.}\label{fig:product_graph}
\end{figure*}

In this section, we analyze the structure and properties of equilibria. We introduce a tool that constructs host graphs with arbitrary size and number of terminals that contain equilibria. In particular, we define a graph product similar to the cartesian product that transfers equilibria from the input graphs to the product graph. This allows us to translate PoA lower bounds between different numbers of terminals. We begin with a description of the construction.

Intuitively, given two host graphs $H_1$ and $H_2$, this operation creates one copy of $H_1$ for each node in $H_2$ and connects the nodes inside those copies to their counterparts in other copies according to $H_2$ such that all edges inside a copy have smaller labels than the edges between copies, filling the gaps with high labeled edges. This leads to temporal paths first travelling inside a local copy of $H_1$ before using edges from $H_2$ to reach the destination copy of $H_1$. Strategy profiles $\s^1$ and $\s^2$ for $H_1$ and $H_2$ are transformed such that the resulting graph contains all time edges inside the copies of $H_1$ that are present in $G(\s^1)$ and all time edges between those copies corresponding to time edges in $G(\s^2)$ but only if the connected nodes correspond to a terminal in $H_1$. See \Cref{fig:product_graph} for an example.

\begin{definition}[graph product]\label{def:graph_product}
    Let $H_1,H_2$ be host graphs with $n_1,n_2$ nodes and $k_1,k_2$ terminals, respectively. We define the \emph{product} of $H_1$ and $H_2$ as a host graph $\product(H_1,H_2)$ such that
    \begin{align*}
        V_{\product(H_1,H_2)}&\coloneqq V_{H_1}\times V_{H_2}\hspace{2.03cm}\text{and}\\
        \forall e=\{(x_1,x_2),(y_1,y_2)\}&\in E_{\product(H_1,H_2)}\colon\\
        \lambda_{\product(H_1,H_2)}(e)&\coloneqq \begin{cases}
            \lambda_{H_1}(\{x_1,y_1\})&\text{ if }x_2=y_2,\\
            \lambda_{H_2}(\{x_2,y_2\})+\lifetime_{H_1}&\text{ if }x_1=y_1,\\
            \{\lifetime_{H_1}+\lifetime_{H_2}+1\}&\text{ else.}
        \end{cases}
    \end{align*}
    Now, we extend this definition to include strategies. Let $\s^1$ and $\s^2$ be strategy profiles for $H_1$ and $H_2$, respectively. We define the \emph{product} of the two strategy profiles and their host graphs $\product(\s^1,H_1,\s^2,H_2)$ as a pair $(\s^\times,H_\times)$ such that $H_\times\coloneqq \product(H_1,H_2)$ and $\s^\times$ is a strategy profile for $H_\times$ with for all $(v_1,v_2)\in V_\times$ we have
    \begin{align*}
        S^\times_{(v_1,v_2)}\coloneqq& \{((x,v_2),(y,v_2),\lambda)\mid \{x,y,\lambda\}\in S^1_{v_1}\}\,\cup\\
        &\{((v_1,x),(v_1,y),\lambda+\lifetime_{H_1})\mid (x,y,\lambda)\in S^2_{v_2}\wedge v_1\in T_{H_1}\}.
    \end{align*}
\end{definition}

Note that, if $\s^1$ and $\s^2$ are local (i.e. every node only buys incident edges), then $\s^\times$ is local, too. Furthermore, if $H_1$ and $H_2$ are simple, $H_\times$ is simple, too.

The next theorem shows how we can translate equilibria from the original two host graphs into an equilibrium for the graph product. 

\begin{restatable}{theorem}{graphProduct}\label{thm:graph_product}
    Let $H_1$ and $H_2$ be host graphs and $\s^1$ and $\s^2$ equilibria of the same type (NE or GE) for a chosen setting (local or global). Further, let $(\s^\times,H_\times)=\product(\s^1,H_1,\s^2,H_2)$. Then $\s^\times$ is an equilibrium for $H_\times$ for the chosen game setting and equilibrium type.
\end{restatable}
\begin{proof}
    The main idea of this proof is to exploit the order of time labels in $H^\times$: The labels of edges inside a copy of $H_1$ are smaller than the labels of edges between copies of $H_1$ that connect nodes corresponding to the same original node in $H_1$ (those edges come from $H_2$). These, in turn, are smaller than edges connecting nodes in different copies of $H_1$ that do not correspond to the same original node in $H_1$. We call the last type of edges \emph{diagonal edges} and show that it is not beneficial for the nodes to buy them. Thus, a best response for some node in $H^\times$ will only buy time edges originating from $H_1$ or $H_2$, not diagonal ones. This best response can then be partitioned into those two kinds of edges which can than be used to construct an improving response either for $\s^1$ or $\s^2$.

    First, we argue that $G(s^\times)$ is a terminal spanner. We observe, that in $G(\s^1)$ and $G(\s^2)$ every node can reach every terminal because otherwise they wouldn't be equilibria. We now argue that this is also true for the resulting product graph. Let $(v_1,v_2)\in V_{H_\times}$ and $(t_1,t_2)\in T_{H_\times}$. Since $t_1$ and $t_2$ are terminals in $H_1$ and $H_2$, respectively, there is a temporal path $v_1,x_1^1,x_1^2,\dots,x_1^p,t_1$ in $G(\s^1)$ and a temporal path $v_2,x_2^1,x_2^2,\dots,x_2^q,t_2$ in $G(\s^2)$. This means there are nodes buying those edges. By construction, the temporal path $(v_1,v_2),(x_1^1,v_2),(x_1^2,v_2),\dots,(x_1^p,v_2),(t_1,v_2)$ and the temporal path $(t_1,v_2),(t_1,x_2^1),(t_1,x_2^2),\dots,(t_1,x_2^q),(t_1,t_2)$ exist in $G(\s^\times)$. Furthermore, by construction, all labels on the second path are larger than the labels on the first path. This means, by concatenating both paths, we obtain a temporal path from $(v_1,v_2)$ to $(t_1,t_2)$. Therefore, all nodes in the product graph can reach all terminals.

    Now, we show that there is no (greedy) improving response for any node in $G(\s^\times)$. Towards a contradiction, let $(v_1,v_2)\in V_\times$ be a node that has an improving response. Let $T^\times$ be its best response. We will create two successively more restrictive new best responses $U^\times$ and $Z^\times$. The latter can then be partitioned to construct either an improving response for $v_1$ for $\s^1$ or an improving response for $v_2$ for $\s^2$.
    
    First, we look at \emph{diagonal} edges $D\coloneqq\{\{(x_1,x_2),(y_1,y_2)\}\in E_\times\mid x_1\neq y_1\wedge x_2\neq y_2\}$. Those are the edges with time label $\lifetime_\times=\lambda^{max}_{H_1}+\lambda^{max}_{H2}+1$. We will abuse notation and use $D$ also as the set of the corresponding time edges. Let $M\subseteq T_{H_\times}$ be the set of terminals that $(v_1,v_2)$ could not reach without diagonal edges. Since diagonal edges have the largest possible time label and no other nodes buy diagonal edges, we have $|M|=|D\cap T^\times|$.  We construct a new strategy $U^\times$ such that
    \begin{align*}
        U^\times\coloneqq T^\times\setminus D \cup \bigcup_{x\in M}\{((v_1,v_2),x,\lambda_x)\}
    \end{align*}
    where $\lambda_x$ is some time label of the edge $\{(v_1,v_2),x\}$.
    We see, that $|U^\times|\le |T^\times|$. Furthermore, $(v_1,v_2)$ is still able to reach all terminals since the replacement of diagonal edges with incident (and still maybe diagonal) time edges doesn't interfere with other temporal paths from $(v_1,v_2)$ due to their high time label. Note, that this step is only necessary in the global setting since in the local setting, all diagonal edges bought by $(v_1,v_2)$ are already incident.

    Next, we replace all remaining diagonal edges by edges inside $(v_1,v_2)$'s copy of $H_1$: 
    \begin{align*}
        Z^\times\coloneqq U^\times\setminus D \cup \bigcup_{\mathclap{\{(v_1,v_2),(x_1,x_2)\}\in U^\times\cap D}}\{((v_1,v_2),(x_1,x_2),\lambda_{x_1,x_2})\}
    \end{align*}
    where $\lambda_{x_1,x_2}$ is again some time label of the edge $\{(v_1,v_2),(x_1,x_2)\}$.
    We again have $|Z^\times|\le |U^\times|$. It remains to show, that $(v_1,v_2)$ can reach every terminal in $G(\s_{-(v_1,v_2)}^\times\cup Z^\times)$. Again, the replacement of diagonal edges only influences the temporal reachability of its endpoints. Let $((v_1,v_2),(x_1,x_2),\lambda)\in U^\times\cap D$ be such a time edge. Then, the edge $\{(v_1,v_2),(x_1,v_2)\}$ is present in $G(\s^\times_{-(v_1,v_2)}\cup Z^\times)$. Also, the temporal path from $(x_1,v_2)$ to $(x_1,x_2)$ from the construction of $\s^\times$ is still present since all time edges on that path are bought by other agents than $(v_1,v_2)$. Lastly, all labels in $\lambda_{H_\times}(\{(v_1,v_2),(x_1,v_2)\})$ are smaller than all time labels on that path. Therefore, there is a temporal path from $(v_1,v_2)$ to $(x_1,x_2)$ in $G(\s^\times_{-(v_1,v_2)}\cup Z^\times)$. Hence, $Z^\times$ is a best response which does not contain any diagonal edges.

    In the final step of this proof, we use $Z^\times$ to construct either an improving response for $v_1$ in $G(\s^1)$ or for $v_2$ in $G(\s^2)$. We observe that due to the construction of $H_\times$ and $\s^\times_{-(v_1,v_2)}\cup Z^\times$ not containing diagonals, that every temporal path from $(v_1,v_2)$ to any terminal $(t_1,t_2)$ first moves inside $(v_1,v_2)$'s copy of $H_1$ and after that moves via edges originating from $H_2$. More concretely, the path has the form $(v_1,v_2),(x_1^1,v_2),(x_1^2,v_2),\dots,(x_1^p,v_2),(t_1,v_2),(t_1,x_2^1),(t_1,x_2^2),\dots$,\\$(t_1,x_2^q), (t_1,t_2)$. Since for $v_1\neq t_1$ the second half of those paths (from $(t_1,v_2)$ to $(t_1,t_2)$ already exist in $G(\s^\times_{-(v_1,v_2)})$ and $Z^\times$ is a best response, we can partition $Z^\times$ into
    \begin{align*}
        Z_1^\times&\coloneqq Z^\times \cap \{((x_1,v_2),(y_1,v_2),\lambda)\mid x_1,y_1\in V_1\}\} \text{ and}\\ Z_2^\times&\coloneqq Z^\times \cap \{((v_1,x_2),(v_1,x_2),\lambda)\mid x_2,y_2\in V_2\}\}.
    \end{align*}
    We use this partition to construct strategies for $v_1$ in $G(\s^1)$ or for $v_2$ in $G(\s^2)$. More precisely, we set
    \begin{align*}
        Z^1&\coloneqq\{(x_1,y_1,\lambda)\mid ((x_1,v_2),(y_1,v_2),\lambda)\in Z_1^\times\} \text{ and}\\
        Z^2&\coloneqq\{(x_2,y_2,\lambda)\mid ((v_1,x_2),(v_1,y_2),\lambda)\in Z_2^\times\}.
    \end{align*}
    
    We now make a case distinction. Case $v_1\notin T_1$: Then, $|Z^1|=|Z^\times_1|=|Z^\times|<|S^\times_{(v_1,v_2)}|=|S^1_{v_1}|$. Now let $t_1\in T_1$ and $t_2\in T_2$ be terminals. Since $Z^\times$ is a best response, there is a temporal path $(v_1,v_2),(x_1^1,v_2),(x_1^2,v_2),\dots,(x_1^p,v_2),(t_1,v_2),(t_1,x_2^1),(t_1,x_2^2),\dots,$\\$(t_1,x_2^q),(t_1,t_2)$ in $G(\s^\times_{-(v_1,v_2)}\cup Z^\times)$ which therefore means that $v_1,x_1^1,x_1^2,\dots,x_1^p,t_1$ is a temporal path in $G(\s^1_{-v_1}\cup Z^1)$. Thus, $Z^1$ is an improving response for $v_1$ contradicting that $\s_1$ is an equilibrium.

    Case $v_1\in T_1$: Then, $|Z^1|+|Z^2|=|Z^\times_1|+|Z^\times_2|=|Z^\times|<|S^\times_{(v_1,v_2)}|=|S^1_{v_1}|+|S^2_{v_2}|$. This means that $|Z^1|<|S^1_{v_1}|$ or $|Z^2|<|S^2_{v_2}|$. If $|Z^1|<|S^1_{v_1}|$, the argument is exactly the same as in the previous case. So let $|Z^2|<|S^2_{v_2}|$ and $t_2\in T_2$. Then there is a temporal path $(v_1,v_2),(v_1,x_2^1),(v_1,x_2^2),\dots,(v_1,x_2^q),(v_1,t_2)$ in $G(\s^\times_{-(v_1,v_2)}\cup Z^\times)$ which means that $v_2,x_2^1,x_2^2,\dots,x_2^q,t_2$ is a temporal path in the graph $G(\s^2_{-v_2}\cup Z^2)$. Therefore, $Z^2$ is an improving response for $v_2$ contradicting that $\s^2$ is an equilibrium.
    
    Since we get a contradiction in both cases, the assumption that there is an improving response for $(v_1,v_2)$ for $\s^\times$ is wrong and therefore $\s^\times$ is an equilibrium for $H_\times$. Note, that all steps were formulated to work for the most general case, namely for NE in the global setting. They however also hold for the local setting since all steps maintain incidency and for GE.
\end{proof}

With the product graph construction, we can get host graphs and equilibria for some combinations of $n$ and $k$. To obtain equilibria for (almost) any combination of $n$ and $k$, we show some useful properties. The following corollary shows that if we we have an equilibrium for $k$ terminals and nodes, we can scale up the example to arbitrarily high number of nodes while keeping the ratio between the time edges in the equilibrium and the number of nodes. That enables us to derive PoA bounds from single examples. The scaled up graphs are created by using the graph product on the original example and a graph with only one terminal.

\begin{restatable}{corollary}{NETerminalsNonTerminals}
    \label{cor:NE_terminals_non_terminals}
    Let $H_1$ be a host graph on $k$ nodes which are all terminals and $\s^1$ be an equilibrium for $H_1$ with $m_1$ time edges. Then, for each $c\in\mathds{N}$, there is a host graph $H$ with $n\coloneqq ck$ nodes and $k$ terminals, and an equilibrium for $H$ containing $m\coloneqq cm_1+(c-1)k$ time edges. Additionally, if $k\ge 2$ and $c\ge 3$, the host graph $H$ contains a spanning tree consisting of edges with label $\lifetime_H$.
\end{restatable}
\begin{proof}
    Let $H_2$ be any host graph that contains exactly $c$ nodes and one terminal $t$. Then every equilibrium for $H_2$ forms a tree containing all nodes such that edge labels are non-decreasing towards $t$. Obviously, such an equilibrium $\s^2$ exists. Let $(\s,H)\coloneqq\product(\s^1, H_1,\s^2,H_2)$ be the graph product of those two host graphs and strategy profiles. By \Cref{thm:graph_product}, $\s$ is an equilibrium. Furthermore, $H$ contains $ck$ nodes and $G(\s)$ contains $cm_1+(c-1)k$ time edges. Lastly, we observe that for $k\ge 2$ and $c\ge 3$ that there is a path of length 2 with edge label $\lifetime_H$ between any pair of nodes consisting of diagonal edges.
\end{proof}

\begin{figure*}
    \centering
    \includegraphics{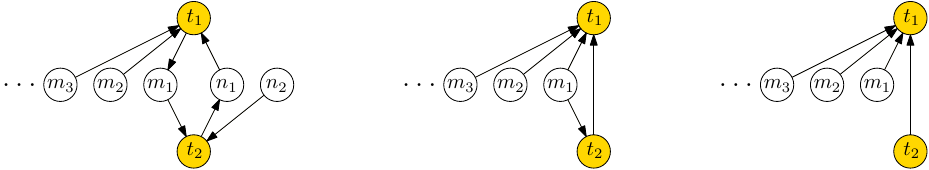}
    \caption{This figure illustrates the equilibrium constructions for the proof of \Cref{thm:NE_two_terminals}. On the left, we have the case where $M\neq\varnothing$ and $N\neq\varnothing$. The middle shows the case where $N=\varnothing$ and $\min\lambda(m_1,t_1)>\min\lambda(t_1,t_2)$ and the right illustrates the case where $N=\varnothing$ and $\min\lambda(m_1,t_1)<\min\lambda(t_1,t_2)$.}
    \label{fig:two_terminals_NE}
\end{figure*}

The following lemma lets us fill the gaps that are left from the previous corollary. \Cref{cor:NE_terminals_non_terminals} lets us create examples with a specific ratio between time edges in the equilibrium and nodes in the graph for arbitrary multiple of the size of a given example. With the following lemma we also get examples for numbers of nodes that are not multiples of the original by adding nodes (and maybe terminals) one by one without changing the number of edges in the equilibrium too much.

\begin{restatable}{lemma}{NEDismountable}\label{lem:NE_dismountable}
    Let $H$ be a host graph on $n$ nodes and $k$ terminals and $\s$ be an equilibrium (Nash or Greedy) creating $m$ time edges in the local or global setting. Then there is
    \begin{compactenum}
        \item a host graph $H_1$ on $n+1$ nodes and $k+1$ terminals and an equilibrium $\s^1$ (of the same type and for the same setting as $\s$) creating at least $m+1$ time edges and
        \item a host graph $H_2$ on $n+1$ nodes and $k$ terminals and an equilibrium $\s^2$ (of the same type and for the same setting as $\s$) creating $m+1$ time edges.
    \end{compactenum}
    Additionally, if $H$ contains a spanning tree consisting of edges with label $\lifetime_H$, then $H_1$, and $H_2$ contain spanning trees consisting of edges with label $\lifetime_{H_1}$ and $\lifetime_{H_2}$, respectively.
\end{restatable}
\begin{proof}
    (2.) Let $x\notin V_H$ and $y\in V_H$. We construct $H_2$ such that $V_{H_2}=V_H\cup\{x\}$, $T_{H_2}=T_H$, and for all $\{u,v\}\in E_{H_2}$
    \begin{align*}
        \lambda_{H_2}(u,v)=\begin{cases}
            \lambda_H(u,v)+1&\text{, if }u\neq x \wedge v\neq x,\\
            \{1\}&\text{, if } \{u,v\}=\{x,y\},\\
            \{\lifetime_H+1\}&\text{, else.}
        \end{cases}
    \end{align*}
   Since $\s$ is an equilibrium, it is easy to see that $\s^2=\s\cup(x,\{(x,y,1)\})$ is an equilibrium for $H_2$, too. Furthermore, a spanning tree with label $\lifetime_H$ in $H$ together with any edge $\{x,v\},v\neq y$ will be a spanning tree with label $\lifetime_H+1$ in $H_2$.
    
    (1.) Let $F$ be the forest containing all edges with label $\lifetime_{G(\s)}$. We make a case distinction:

    Case 1: $F$ is a tree and contains all terminals from $H$. Then $G(\s)$ is a tree. For $H_1$, we can simply take the host graph with $n+1$ nodes, $k+1$ terminals and all edges having label 1.

    Case 2: There is a terminal from $H$ that is not in $F$ or a $F$ is not a tree. In both cases, we can find an (inclusion maximal) tree $B$ in $F$ and a terminal $t$ not in $B$. There also is a terminal $t_B$ in $B$ since otherwise all edges in $B$ would not be beneficial for any node. Lastly, there is a node $a$ in $B$ that does not buy any edge in $B$.
    
    Let $b\neq a$ be another node in $B$ and $x\notin V_H$. We construct $H_1$ such that $V_{H_1}=V_H\cup\{x\}$, $T_{H_1}=T_H\cup\{x\}$, and for all $\{u,v\}\in E_{H_1}$
    \begin{align*}
        \lambda_{H_2}(u,v)=\begin{cases}
            \lambda_H(u,v)+1&\text{, if }u\neq x \wedge v\neq x,\\
            \{1\}&\text{, if } \{u,v\}=\{x,a\},\\
            \{\lifetime_H+1\}&\text{, else.}
        \end{cases}
    \end{align*}
    Construct $\s^1$ such that $S^1_x=\{(x,a,1)\}$, $S^1_b=S_b\cup\{(x,b,\lifetime_H+1)\}$ and $\forall v\in V_{H}\setminus\{b\}\colon S^1_v=S_v$.

    We see, that $G(\s^1)$ contains $m+2$ edges. We now argue, that $\s^1$ is an equilibrium. We first observe that $G(\s^1)$ is a terminal spanner since every node can reach every terminal from $H$ just like in $G(\s)$ and can also reach $x$ by first going to $t_B$ and then to $x$ through $B$ and via $\{x,b\}$.

    Finally, we argue that no node can buy a smaller edge set. $x$ cannot reach $t$ via the edge $\{x,b\}$ because $t$ is not in $B$. Therefore, $x$ does not want to remove the edge $\{x,a\}$ which indeed ensures reachability for $x$ to every terminal since $a$ can reach every terminal. The only other node that could benefit from $\{x,a\}$ is $a$ but $a$ could only use it to reach nodes in $B$. Since $a$ doesn't buy edges in $B$, $a$ cannot use this to improve her strategy. Since $b$ has no other option to reach $x$ besides buying $\{b,x\}$ she cannot remove this edge. Any improving response by any node $v\neq x$ would also be an improving response for $\s$ on $H$. Thus, $\s^1$ is an equilibrium.

    Note, that a spanning tree with label $\lifetime_H$ in $H$ together with any edge $\{x,v\},v\neq a, v\neq b$ will be a spanning tree with label $\lifetime_H+1$ in $H_1$.
\end{proof}

The previous lemmas give us a way to create graphs with equilibria of specific sizes. In the next part, we analyze for which kind of host graphs equilibria exist.

The first observation is that when the host graph only contains two distinct labels than an equilibrium in the form of a spanning tree exists. This follows the same way as the same statement for the local setting in \cite{tempNCG_ijcai}.
\begin{restatable}{corollary}{NELifetimeTwo}\label{cor:NE_lifetime_two}
    Let $H$ be a host graph with $\lifetime=2$. Then there is a NE $\s$ for $H$. Furthermore, $G(\s)$ is a spanning tree.
\end{restatable}
\begin{proof}
    The proof follows the same arguments as in \cite{tempNCG_ijcai} Theorem 5.
\end{proof}

Interestingly, we can also show the existence of equilibria if we restrict the host graph to only two terminals.
\begin{restatable}{theorem}{NETwoTerminals}\label{thm:NE_two_terminals}
    Let $H$ be a host graph with $k=2$ terminals. Then there is a NE $\s$ for the local and the global setting containing at most $n$ (time) edges.
\end{restatable}
\begin{proof}
    For simplicity, we assume distinct edge labels in the host graph. When edge labels are not distinct that might lead to one of the edges of our construction being dropped. Other than that the proof works the same. Let $T=\{t_1,t_2\}$ be the set of terminals. We partition all remaining nodes into
    \begin{align*}
        M&\coloneqq\{v\in V\setminus T\mid \min\lambda(t_1,v)< \min\lambda(v,t_2)\}\qquad\text{and}\\
        N&\coloneqq\{v\in V\setminus T\mid \min\lambda(t_1,v)> \min\lambda(v,t_2)\}.
    \end{align*}
    Let further $m_1,\dots,m_p\in M$ be the nodes in $M$ sorted descendingly by $\min\lambda(t_1,m_i)$ and $n_1,\dots,n_q\in N$ be the nodes in $N$ sorted descendingly by $\min\lambda(n_i,t_2)$.

    If $M$ and $N$ are not empty, we construct the strategy profile $\s$ such that for all $v\in V$
    \begin{align*}
        S_v\coloneqq\begin{cases}
            \{(t_1,m_1,\min\lambda(t_1,m_1))\}&\text{, if } v=t_1\\
            \{(m_1,t_2,\min\lambda(m_1,t_2))\}&\text{, if } v=m_1\\
            \{(t_2,n_1,\min\lambda(t_2,n_1))\}&\text{, if } v=t_2\\
            \{(n_1,t_1,\min\lambda(n_1,t_1))\}&\text{, if } v=n_1\\
            \{(v,t_1,\min\lambda(v,t_1))\}&\text{, if } v\in M\setminus\{m_1\}\\
            \{(v,t_2,\min\lambda(v,t_2))\}&\text{, if } v\in N\setminus\{n_1\}.
        \end{cases}
    \end{align*}
    See \Cref{fig:two_terminals_NE} for an illustration of $G(\s)$.

    We see, that every node is able to reach both $t_1$ and $t_2$ since both terminals can reach each other and all other nodes arrive at their adjacent terminal before. Furthermore, every node buys exactly one (time) edge. Therefore, an improving response can only be the removal of that edge. This is trivially not beneficial for all nodes except $t_1, t_2, m_1$ and $n_1$ because this would disconnect them. If $t_1$ or $m_1$ would remove their edge, they wouldn't be able to reach $t_2$. The same holds for $t_2$ and $n_1$ regarding $t_1$. Thus, $\s$ is indeed a NE.

    We now consider the case where $M$ or $N$ is empty. Let, without loss of generality, $N=\varnothing$. Then $n_1$ does not exist. We now make a case distinction:

    Case $\min\lambda(m_1,t_1)<\min\lambda(t_1,t_2)$: Here, we set for all $v\in V$
    \begin{align*}
        S_v\coloneqq\begin{cases}
            \varnothing&\text{, if } v=t_1\\
            \{(m_1,t_1,\min\lambda(m_1,t_1))\}&\text{, if } v=m_1\\
            \{(t_2,t_1,\min\lambda(t_2,t_1))\}&\text{, if } v=t_2\\
            \{(v,t_1,\min\lambda(v,t_1))\}&\text{, if } v\in M\setminus\{m_1\}.
        \end{cases}
    \end{align*}
    This yields a star where every node can reach both $t_1$ and $t_2$ and any edge removal would result in disconnecting the removing node from the rest of the graph. Therefore, $\s$ is a NE.

    Case $\min\lambda(m_1,t_1)>\min\lambda(t_1,t_2)$: Here, we set for all $v\in V$
    \begin{align*}
        S_v\coloneqq\begin{cases}
            \varnothing&\text{, if } v=t_1\\
            \{(m_1,t_1,\min\lambda(m_1,t_1)),\\\qquad\quad(m_1,t_2,\min\lambda(m_1,t_2))\}&\text{, if } v=m_1\\
            \{(t_2,t_1,\min\lambda(t_2,t_1))\}&\text{, if } v=t_2\\
            \{(v,t_1,\min\lambda(v,t_1))\}&\text{, if } v\in M\setminus\{m_1\}.
        \end{cases}
    \end{align*}
    It is again easy to see, that each node can reach every terminal. Also, no node wants to sell any of their edges. But since $m_1$ is now buying two edges, we need to consider the case where $m_1$ sells both edges and buys a single edge instead. In order to still be able to reach $t_1$ and $t_2$, it needs to be an edge to some node $m_i$ such that $\min\lambda(m_i,t_1)<\min\lambda(t_1,t_2)$. If there is such a node $m_i$ we change $S_1$ to $\{(m_1,m_i,\min\lambda(m_1,m_1))\}$ and replace the role of $m_1$ by $m_2$, we can now make the same case distinction again:
    If $\min\lambda(m_2,t_1)<\min\lambda(t_1,t_2)$, setting $S_{m_2}=\{(m_2,t_1,\min\lambda(m_2,t_1))\}$ yields an equilibrium. If $\min\lambda(m_2,t_1)>\min\lambda(t_1,t_2)$ we set the strategy of $m_2$ to $S_{m_2}=\{(m_2,t_1,\min\lambda(m_2,t_1)),(m_2,t_2,\min\lambda(m_2,t_2))\}$ and repeat. This process will stop eventually by going into the first case leaving us with an NE. Note, that we didn't assume incident edges but the resulting strategy profile only contains incident edges. Therefore, this result holds for both the local and the global setting.
\end{proof}

\subsection{Local Edge-Buying $k$-Terminal TNCG}\label{subsec:local}
In this section, we analyze properties of equilibria for the local setting.

We prove an upper bound on the number of edges in an equilibrium state dependent on the number of terminals $k$. The proof is a modified version of \cite{tempNCG_ijcai}[Definition 7 - Theorem 9]. The key difference is that their proof only holds for $k=n$ and one time label per edge. For the sake of clarity, we reprove their lemmas for our modified version by slightly adjusting their proofs.

Note that in an equilibrium in the local setting, the agents only buy at most one time edge per edge in the graph. That is the case because everything the incident nodes can reach over the later time edge, they can also reach by replacing it with the earlier one, so they have an improving move by removing the later time edge from their strategy. Hence, all equilibria in this setting are simple temporal graphs. We call a strategy profile that creates a simple temporal graph a simple strategy profile. For ease of notation, we talk about buying edges instead of time edges in simple strategy profiles and use $\lambda(e)$ to denote the label of the corresponding label of the bought time edge. We start by introducing the concept of \emph{necessary edges} which we then use to characterize a structure that cannot appear in an equilibrium. 
\begin{definition}[necessary edge]
    Let $H$ be a host graph with $n$ agents and $k$ terminals, $\s$ a simple strategy profile, and $G\coloneqq G(\s)$. For each edge $e=\{u,v\}\in G$ that $u$ buys, we define
    $$A_G(e)\coloneqq \big\{x\in T_H\mid x\in \reach{G}(u)\wedge x\notin\reach{G-e}(u)\big\}.$$
    We say that $e$ is \emph{necessary} for $u$ to reach the terminals in $A_G(e)$.
\end{definition}
Note, that if $\s$ is an equilibrium, $A_G(e)\neq\varnothing$ for all $e\in E_G$.

Using this definition, we characterize a structure that cannot appear in any simple strategy profile.
\begin{lemma}\label{lem:forbidden_structure}
    Let $H$ be a host graph, $\s$ be simple a strategy profile, and $G\coloneqq G(\s)$. There cannot be nodes $z,u_1,u_2\in V$ and $x,y \in T_H$ and distinct edges $e_{1x}, e_{1y}, e_{2x}, e_{2y}\in E_G$ such that for $i\in\{1,2\}$ and $j\in\{x,y\}$
    \begin{compactenum}
        \item $\{z,u_i\}\in E_{G(\s)}\setminus\{e_{ij}\}$;
        \item $e_{ij}$ is bought by $u_i$;
        \item $j\in A_G(e_{ij})$; 
        \item $\lambda(\{z,u_i\})\leq\lambda(e_{ij})$. \hfill (See~\Cref{fig:forbidden_structure}.)
    \end{compactenum}
    
\end{lemma}
\begin{figure}
    \centering
    \includegraphics[scale=1]{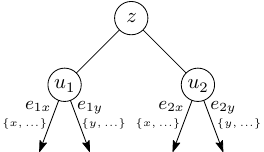}
    \caption{A forbidden structure in a strategy profile. The node $z$ has two neighbors $u_1$ and $u_2$ that both buy two distinct edges that they need to reach the nodes $x$ and $y$ respectively. For both of them, the two needed edges have at least a label as high as their edge to $z$.}
    \label{fig:forbidden_structure}
\end{figure}
\begin{proof}
    Towards a contradiction, suppose that there are nodes $z,u_1,u_2\in V$ and $x,y \in T_H$ and edges $e_{1x}, e_{1y}, e_{2x}, e_{2y}\in E_G$ as defined above. W.l.o.g. $\lambda(\{z,u_1\})\le\lambda(\{z,u_2\})$ and $\lambda(e_{1x})\le\lambda(e_{1y})$. Any temporal path from $u_2$ to $x$ starts with $e_{2x}$ (since $x\in A_G(e_{2x})$) and has to use $e_{1x}$. Otherwise, $u_1$ could reach $x$ by using $\{z,u_1\}$, $\{z,u_2\}$ and then the path from $u_2$ to $x$ without needing $e_{1x}$ which contradicts $x\in A_G(e_{1x})$. The same holds for $y$ instead of $x$.
    
    Therefore, there is a temporal path $P$ from $u_2$ to $u_1$, which starts with $e_{2x}$ and arrives at $u_1$ no later than $\lambda(e_{1x})\le\lambda(e_{1y})$. This means that $u_2$ does not need $e_{2y}$ to reach $y$ since it can use $P$ to get to $u_1$ and travel to $y$ from there. This contradicts $y\in A_G(e_{2y})$.
\end{proof}

We now show that simple temporal graphs with at least $n\sqrt{6k}+n$ edges contain a node that has a lot of neighbors that all need some late edges. Note that we consider directed temporal graphs in the proof. The direction indicates that the source buys the edge in the given strategy.

\begin{lemma}\label{lem:large_node_exists}
    Let $G$ be a simple directed temporal graph with $n$ nodes and at least $n\sqrt{6k}+n$ edges. Then there is a node $z\in V_G$ and a set $M\subset V_G$ such that
    \begin{compactenum}
        \item $|M| = \lceil\frac{1}{3}\sqrt{6k}\rceil$;
        \item $(u,z) \in E_G$, for every $u \in M$;
        \item each $u\in M$ has at least $\frac{2}{3}\sqrt{6k}$ outgoing edges $e=(u,v)\in E_G$ with $z\neq v$ and $\lambda_G((u,z))\leq\lambda_G(e)$.
    \end{compactenum}
\end{lemma}
\begin{proof}
    Consider the graph $G'$ where for each node $v$ we remove the $\lceil\frac{2}{3}\sqrt{6k}\rceil$ outgoing edges with the largest labels (break ties arbitrarily).
    If $v$ has less outgoing edges, we just remove all of them. Now, $G'$ has at least $n\sqrt{6k}+n-n\lceil\frac{2}{3}\sqrt{6k}\rceil\ge n \frac{1}{3}\sqrt{6k}$ edges. By the pigeonhole principle, there is a node $z$ with at least $\lceil\frac{1}{3}\sqrt{6k}\rceil$ incoming edges. Let $M$ be a set of $\lceil\frac{1}{3}\sqrt{6k}\rceil$ neighbors $u$ of $z$ in $G'$ that have a directed edge $e'=(u,z)\in E_{G'}$ towards $z$. By construction of $G'$, each $u\in M$ has at least $\frac{2}{3}\sqrt{6k}$ outgoing edges $e=(u,v)\in E_G$ such that $v \neq z$ and $\lambda_G(e')\leq\lambda_G(e)$  (see~\Cref{fig:graph_including_forbidden_structure}).
\end{proof}

To conclude our bound on equilibria, we show that the structure from \Cref{lem:large_node_exists} implies the existence of the forbidden structure from~\Cref{fig:forbidden_structure}. Therefore, simple graphs with at least $n\sqrt{6k}+n$ edges must contain unnecessary edges, giving us a bound on the number of edges in an equilibrium.

\begin{restatable}{theorem}{denseNotGE}\label{thm:dense_not_ge}
    Let $H$ be a host graph with $|V_H|=n$ agents and $k$ terminals and let $\s$ be a strategy profile in the local setting. If $G\coloneqq G(\s)$ contains at least $\sqrt{6k}n+n$ time edges, then $G$ is not a GE.
\end{restatable}
\begin{proof}
    Towards a contradiction, suppose that $G$ is a GE and $|M_G|\ge n\sqrt{6k}+n$. As $G$ is an equilibrium, it has to be simple and $|E_G| = |M_G|$. As shown in \Cref{lem:large_node_exists} (the direction of the graph indicates that the source of the edge buys it in $\s$), there is a node $z \in V_G$ and a set $M\subset V_G$ such that
    \begin{compactenum}
        \item $|M| = \lceil\frac{1}{3}\sqrt{6k}\rceil$;
        \item $\{u,z\} \in E_G$, for every $u \in M$;
        \item Each $u\in M$ has a set $E_u\subseteq S_u$ of at least $\frac{2}{3}\sqrt{6k}$ bought edges $(u,v)$ with $z\neq v$ and $\lambda_G((u,z))\le\lambda((u,v))$.
    \end{compactenum}
    See \Cref{fig:graph_including_forbidden_structure} for an illustration.
    \begin{figure}[t]
    \centering
    \includegraphics[scale=1]{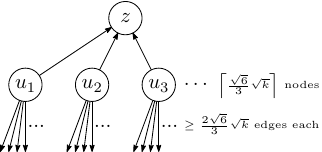}
    \caption{A structure that always appears in a directed temporal graph with at least $n\sqrt{6k}+n$ edges. A node $z$ exists with $\lceil\frac{\sqrt{6}}{3}\sqrt{k}\rceil$ neighbors via in-edges that each have at least $\frac{2\sqrt{6}}{3}\sqrt{k}$ out-edges with a label that is at least as high as the label of their edge to~$z$.}
    \label{fig:graph_including_forbidden_structure}
\end{figure}

    For each edge $e\in E_u$, let $a_e\in A_G(e)$ be a representative of $A_G(e)$. Note that $A_G(e)\neq\varnothing$ because $G$ is a GE, so those representatives always exist.
    For each $u\in M$, we define
    \begin{equation*}
        D_u\coloneqq\bigcup_{e\in E_u}\{a_e\}.
    \end{equation*}
    Intuitively, $D_u$ contains terminals that $z$ can reach by going over $u$ and that $u$ needs to buy an edge for.
    
    We see that the forbidden structure from~\Cref{lem:forbidden_structure} appears if there are two nodes $u,v\in M$ such that $|D_u\cap D_v|\ge 2$. We can therefore assume $|D_u\cap D_v|\leq 1$ for all $u,v \in M$. Also, for $e,e'\in E_u$ we have $A_G(e)\cap A_G(e')=\varnothing$ since there cannot be two edges that are necessary for $u$ to reach the same node. From this, we get $|D_u|\ge |E_u|\ge\frac{2}{3}\sqrt{6k}$.
    
    Using the inclusion-exclusion principle, we get
    \begin{align*}
        \left|\bigcup_{u\in M}D_u\right|&\ge\sum_{u\in M}|D_u|-\smashoperator[r]{\sum_{\{u,v\}\subseteq M,u\neq v}}|D_u\cap D_v|\\
        &\ge \left\lceil\frac{1}{3}\sqrt{6k}\right\rceil\frac{2}{3}\sqrt{6k}-\frac{1}{2}\left\lceil\frac{1}{3}\sqrt{6k}\right\rceil\left(\left\lceil\frac{1}{3}\sqrt{6k}\right\rceil-1\right)\\
        &> \left\lceil\frac{1}{3}\sqrt{6k}\right\rceil\frac{2}{3}\sqrt{6k}-\frac{1}{2}\left\lceil\frac{1}{3}\sqrt{6k}\right\rceil\frac{1}{3}\sqrt{6k}\\
        &= \left\lceil\frac{1}{3}\sqrt{6k}\right\rceil \frac{1}{2}\sqrt{6k} \geq k.
    \end{align*}
    
    This is a contradiction since each set $D_u$ only contains terminals and $H$ has only $k$ terminals.
\end{proof}

\subsection{Global Edge-Buying $k$-Terminal TNCG}
In this section, we analyze properties of equilibria for the global setting. We start by showing some differences between the structural properties of Greedy and Nash Equilibria. We finish with our main result for this section which shows that equilibria in the global and local setting are incomparable.

The set of inclusion minimal temporal spanners and the set of Greedy Equilibria coincide.
\begin{restatable}{theorem}{GETemporalSpanners}
    \label{thm:GETemporalSpanners}
    Let $H$ be a host graph in the global setting.
    \begin{compactenum}[(i)]
        \item For every GE $\s$, the graph $G(\s)$ is an inclusion minimal terminal spanner of $H$.
        \item For every inclusion minimal terminal spanner $G$ of $H$, there is a GE $\s$ with $G(\s)=G$.
    \end{compactenum}
\end{restatable}
\begin{proof}
    To prove (i), let $\s$ be a Greedy Equilibrium. If $G(\s)$ is not a terminal spanner, there is an agent that cannot reach one of the terminals. Buying any time edge directly to that terminal is an improving move. So $G(\s)$ has to be a terminal spanner. Assume, $G(\s)$ is not an inclusion minimal terminal spanner for $H$. Then there is a time edge $(e,\lambda)\in \timeEdges_{G(\s)}$ such that $G(\s)-(e,\lambda)$ is a terminal spanner. Any agent buying $(e,\lambda)$ can remove it from its strategy and still reach every terminal. Therefore, this is a greedy improving response and thus $\s$ not a Greedy Equilibrium. This contradicts the assumption.

    For proving (ii), let $G$ be an inclusion minimal terminal spanner for $H$. For every time edge $e\in \timeEdges_G$, there is a pair of nodes $u,t$ with $t\in T_H$, s.t $u$ cannot reach $t$ in $G-e$. We set $v_e\coloneqq u$. If there are multiple such pairs for $e$, we arbitrarily pick one. Now, we construct a strategy profile $\s$ such that for all nodes $u\in V$
    \begin{align*}
        S_u\coloneqq \{e\in \timeEdges_G\mid v_e=u\}.
    \end{align*}
    This means, every time edge $e\in \timeEdges_G$ is bought by exactly one node, namely $v_e$. Therefore, $G(\s)=G$ which also means that $G(\s)$ is a terminal spanner. Thus, no node wants to add a time edge. Similarly, no node $u\in V$ can remove a time edge since this would result in $u$ not reaching some terminal. Therefore, $\s$ is a GE.
\end{proof}

\begin{restatable}{corollary}{GEexist}\label{cor:GEexists}
    For every host graph  $H$ a GE exists in the global setting.
\end{restatable}
\begin{proof}
    This follows directly from \Cref{thm:GETemporalSpanners}.
\end{proof}

For NE, equilibria and minimal temporal spanners do not coincide. We show that by providing a minimal temporal spanner that does not admit an equilibrium no matter how the edges are assigned.

\begin{figure}
    \centering
    \begin{tikzpicture}[on grid=true, node distance=1.5cm and 1.5cm]
    	\node[nod] (v1) at (135:1.5){$v_1$};
    	\node[nod] (v2) at (45:1.5) {$v_2$};
    	\node[nod] (v3) at (315:1.5){$v_3$};
    	\node[nod] (v4) at (225:1.5){$v_4$};
    	\path[dashed, red]
    	   \edo(v1,v3,1, above);
    	\path[blue]
    	   \edo(v2,v3,4;v_3,right)
    	   \edo(v1,v2,5;v_3, above)
    	   \edo(v1,v4,2;v_1,left)
          \edo(v2,v4,2;v_2,below)
          \edo(v3,v4,3;v_3,below)
            ;
     
    \end{tikzpicture}\hspace{1cm}
    \caption{Simple temporal clique with a given minimal temporal spanner (blue edges). The numbers represent the labels of the edges. For blue edges it is also given which nodes could not reach all other nodes anymore when the edge is removed from the spanner.}
    \label{fig:NETemporalSpanner}
\end{figure}
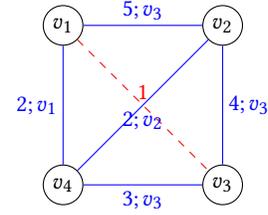

\begin{restatable}{lemma}{NETemporalSpanners}
    There exist simple host graphs $H$ with $k=n$ terminals and inclusion minimal temporal spanners $G$ for $H$ such that there is no NE $s$ with $G(s) = G$ in the global setting.
\end{restatable}
\begin{proof}
    Consider the simple host graph $H$ and the inclusion minimal temporal spanner $G$ of $H$ from \Cref{fig:NETemporalSpanner}. Now assume there is a Nash Equilibrium $s$ with $G(s) = G$. In this equilibrium every (time) edge of $G$ is assigned to one of the agents. If any edge is assigned to a different agent than the one given in \Cref{fig:NETemporalSpanner}, this agent can remove the edge from its strategy and still reach every other agent. Therefore, the given assignment is the only possible Nash Equilibrium. However, in this assignment, agent $v_3$ can remove all edges from its strategy and add edge $(v_1,v_3)$ to maintain reachability to every node. That is an improving response as it reduces its strategy by two edges. This is a contradiction to $s$ being a Nash equilibrium, which concludes the proof.
\end{proof}

The next theorem compares local and global equilibria, showing that those settings are really different and no setting is a generalization of the other in terms of the generated equilibrium graphs.

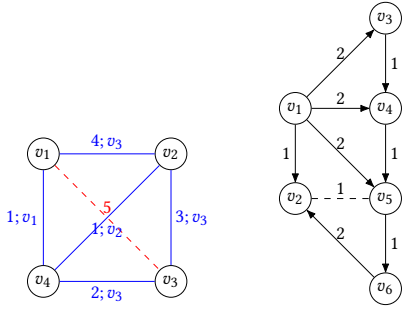
\begin{figure}
    \centering
    \scalebox{0.8}{\begin{tikzpicture}[on grid=true, node distance=1.5cm and 1.5cm]
    	\node[nod] (v1) at (135:1.5){$v_1$};
    	\node[nod] (v2) at (45:1.5) {$v_2$};
    	\node[nod] (v3) at (315:1.5){$v_3$};
    	\node[nod] (v4) at (225:1.5){$v_4$};
    	\path[dashed, red]
    	   \edo(v1,v3,5, above);
    	\path[blue]
    	   \edo(v2,v3,3;v_3,right)
    	   \edo(v1,v2,4;v_3, above)
    	   \edo(v1,v4,1;v_1,left)
          \edo(v2,v4,1;v_2,below)
          \edo(v3,v4,2;v_3,below)
            ;
     
    \end{tikzpicture}\hspace{1cm}
    \begin{tikzpicture}[on grid=true, node distance=1.5cm and 1.5cm]
    	\node[nod] (v1) {$v_1$};
    	\node[nod] (v2) at (0,-1.5) {$v_2$};
    	\node[nod] (v3) at (1.5,1.5){$v_3$};
    	\node[nod] (v4) at (1.5,0){$v_4$};
        \node[nod] (v5) at (1.5,-1.5){$v_5$};
    	\node[nod] (v6) at (1.5,-3){$v_6$};
    	\path[dashed]
    	\edo(v2,v5,1,above);
     
        \path[-{Latex[round]}]
    	\ed(v1,v3,2)
    	\ed(v1,v4,2)
        \ed(v1,v5,2)
        \ed(v6,v2,2)
        \edo(v1,v2,1,left)
        \edo(v3,v4,1,right)
        \edo(v4,v5,1,right)
        \edo(v5,v6,1,right);
    \end{tikzpicture}}
    \caption{Simple temporal cliques with given global or local equilibrium respectively. The numbers on the edges indicate the time labels and all non depicted edges have label 3. In the left graph the blue edges form a global equilibrium when bought by the indicated nodes. In the right graph the directions of the edges indicate that the source buys this edge in a local Nash equilibrium.}
    \label{fig:NEIncomparable}
\end{figure}

\begin{restatable}{theorem}{NEIncomparable}
    Let for a host graph $H$ with $k=n$ terminals, $\gne_H$ be the set of Nash equilibria in the global setting on $H$ and let $\lne_H$ be the set of Nash equilibria in the local setting. Then the set of graphs defined by those equilibria are incomparable. In particular:
    \begin{compactitem}
    \item There exists a simple host graph $H_1$ and a global equilibrium $A_g \in \gne_{H_1}$ such that for all $B_l \in \lne_{H_1}$ holds $G(A_g) \neq G(B_l)$.
    \item There exists a simple host graph $H_2$ and a local equilibrium $A_l \in \lne_{H_2}$ such that for all $B_g \in \gne_{H_2}$ holds $G(A_l) \neq G(B_g)$.
    \end{compactitem}
\end{restatable}
\begin{proof}
    Let $H_1$ be the left graph from \Cref{fig:NEIncomparable} and let $A_g$ be the global equilibrium depicted. It is an equilibrium because it is temporally connected and none of $v_1$, $v_2$, $v_4$ can improve because they all have at most one (time) edge in their strategy and cannot remove it without breaking connectivity. The agent $v_3$ cannot improve because the only other edge it can buy is $(v_1,v_3)$ which still needs both $(v_2,v_3)$ and $(v_3,v_4)$ to keep connectivity to all other nodes.

    There exists no local equilibrium $B_l \in \lne_{H_1}$ with $G(A_g) = G(B_l)$. The reason for that is, that in the local setting the edge $(v_1,v_2)$ has to be in the strategy of either $v_1$ or $v_2$ in order to create $G(A_g)$. However, neither of them needs this edge for reachability in $G(A_g)$, so whoever has it in its strategy has an improving response by removing this edge.

    Let $H_2$ be the right graph from \Cref{fig:NEIncomparable} and let $A_l$ be the local equilibrium depicted. It is an equilibrium because it is temporally connected and all nodes but $v_1$ cannot improve because they all have at most one edge in their strategy and cannot remove it without breaking connectivity. The agent $v_1$ cannot improve because it cannot remove any other edge from its strategy (even while adding new edges) without losing connectivity to the target of the respective edge.

    There exists no global equilibrium $B_g \in \gne_{H_2}$ with $G(A_l) = G(B_g)$. The reason for that is, that in the global setting the edges $(v_1,v_3)$, $(v_1,v_4)$ and $(v_1,v_5)$ have to be in the strategy of either $v_1$ or $v_2$ in order to create $G(A_l)$ and to form an equilibrium. All other agents could remove any one of those edges while keeping connectivity to all other nodes. Hence, by pigeonhole principle, either $v_1$ or $v_2$ has to have two of those edges in their strategy. However, that gives this agent an improving response by removing those two edges and adding the edge $(v_2,v_5)$ which creates a spanning tree with label 1 and therefore ensures all pair reachability.
\end{proof}

\section{Efficiency of Equilibria}\label{sec:poa}
In this section, we analyze the efficiency of equilibria by comparing their social cost to the social cost of socially optimal networks. In particular, we derive several bounds on the Price of Anarchy and Price of Stability. We start by bounding the size of a social optimum. The bound follows directly from the fact that social optima are minimum temporal spanners and the $\mathcal{O}(n\log(n))$ upper bound on the size of minimum temporal spanners on temporal cliques from \citet{CasteigtsPS21}.

\begin{restatable}{corollary}{theoremUBsocopt}\label{cor:theoremUBsocopt}
    Let $H$ be a host graph with $|V_H|=n$ agents and $\s^*$ be a social optimum for $H$. Then $\SC{\s^*}{H}\in\mathcal{O}(n\log(n))$.
\end{restatable}
\begin{proof}
    $\opt=G(\s^*_H)$ is a minimum terminal spanner of $H$, i.e., the terminal spanner that contains the least edges. Let $H'$ be a subgraph of $H$ that contains exactly one label for each edge of $H$. It is easy to see, that every temporal spanner of $H'$ is also a terminal spanner of $H$. Thus, $G(\s^*_H)$ contains at most as many time edges as a minimum temporal spanner for $H'$. Casteigts, Peters, and Schoeters \cite{CasteigtsPS21} proved that minimum temporal spanners of simple complete temporal graphs contain at most $\mathcal{O}(n\log n)$ edges. Therefore, we have $\SC{\opt}{H}=|\timeEdges_\opt|\in\mathcal{O}(n\log(n))$
\end{proof}

For GE in the global setting we exactly know the PoS as the social optimum is also an equilibrium. 

\begin{restatable}{corollary}{GEPoS}\label{cor:GEPoS}
    $\pos_\GE^\gl(n,k)=1$.
\end{restatable}
\begin{proof}
    Let $H$ be a temporal host graph and $G$ be a minimum terminal spanner for $H$. $G$ is also an inclusion minimal terminal spanner for $H$ and with \Cref{thm:GETemporalSpanners} there is a strategy profile $\s$ in GE with $G(s) = G$. The strategy profile $\s$ is also a social optimum. Therefore, $\pos_\GE^\gl=1$.
\end{proof}

The rest of the section analyzes the Price of Anarchy in multiple settings. We first give an upper bound on the PoA in all settings based on the maximum lifetime. It follows from the fact that there will not be a cycle of any label in an equilibrium which means that equilibria can contain at most $\lifetime (n-1)$ time edges.

\begin{restatable}{theorem}{PoALifetime}\label{thm:PoALifetime}
    For host graphs with a maximum lifetime of $\lifetime$, it holds that $\poa\in \mathcal{O}(\lifetime)$.
\end{restatable}
\begin{proof}
    Let $H$ be a temporal host graph with $n$ nodes and a lifetime of $\lifetime$ and let $A$ be an equilibrium for $H$. Assume that $A$ has at least $\lifetime(n-1)+1$ time edges. Then, by pigeonhole principle, there has to be a label $l$ that appears on at least $n$ time edges of $A$. Therefore there exists a cycle in $A$ that contains only time edges with label $l$. Let $e$ be a time edge of this cycle and let $v$ be the agent that has $e$ in its strategy.  Agent $v$ now has an improving response by deleting $e$ from its strategy as everything it could reach using $e$, it can now reach by using the rest of the cycle instead. That is a contradiction to $A$ being an equilibrium, so there cannot be an equilibrium with at least $\lifetime(n-1)+1$ time edges.
\end{proof}

In the local setting we bound the PoA from both directions dependant on $k$. For the lower bound we first use the graph product from \Cref{def:graph_product} to create graphs with $k$ nodes and terminals that contain spanning trees with a single label and equilibria that are hyper cubes with $\Theta(k \log k)$ time edges. We then use \Cref{cor:NE_terminals_non_terminals} and \Cref{lem:NE_dismountable} to blow up the examples to arbitrarily large graphs that still have $\log k$ times as many time edges in an equilibrium than in the social optimum.

\begin{restatable}{theorem}{PoALBLocal}
    \label{thm:PoA_lower_bound_local}
    $\poa^\lo(n,k)\in\Omega(\log k)$.
\end{restatable}
\begin{proof}
    Let $k=2^d$ for some $d\in\N$. We now construct a sequence of host graphs $H_1,\dots, H_d$ by repeatedly applying the product graph result. Let $H_1$ be the host graph containing $n_1=2$ nodes, $k_1=2$ terminals and exactly one label on the only edge. It is easy to see, that there is a NE $\s^1$: One of the nodes buys the only existing time edge. We now define $(\s^{i+1},H_{i+1})\coloneqq(\s^i,H_i,\s^1,H_1)$ for all $1<i<d$. By \Cref{thm:graph_product} $\s^i$ is a NE for $H_i$. We can see, that $G(\s^i)$ is an $i$-dimensional hypercube. Thus, $G(\s^d)$ contains $\Theta(k\log k)$ time edges.

    Using \Cref{cor:NE_terminals_non_terminals} (if $
    \frac{n}{k}\ge 3$) and \Cref{lem:NE_dismountable}, we obtain a host graph $H$ and a NE $\s$ containing $\Theta(k\log k\cdot\frac{n}{k})$ time edges. For $k\ge 2$, the host graph contains a spanning tree consisting of edges with label $\lifetime$. Thus, $\opt$ contains only $n-1$ time edges. We obtain
    \begin{align*}
        \poa^\lo\ge\frac{\socialcost_{H_d}(\s^d)}{\socialcost_{H_d}(\opt)}\in\Omega\left(\frac{k\log k\cdot\frac{n}{k}}{n-1}\right)=\Omega(\log k).&\qedhere
    \end{align*}
\end{proof}

The $\mathcal{O}(\sqrt k)$ upper bound follows directly from \Cref{thm:dense_not_ge}.

\begin{restatable}{corollary}{PoAUBLocal}
    \label{thm:PoA_upper_bound_local}
    $\poa^\lo(n,k)\in\mathcal{O}(\sqrt k)$.
\end{restatable}
\begin{proof}
    By \Cref{thm:dense_not_ge}, the number of time edges in a GE (and thus in a NE, too) is upper bounded by $\Omega(n\sqrt{k})$. The number of time edges in an optimum is lower bounded by $n-1$. Thus,
    \begin{align*}
        \poa^\lo\in\mathcal{O}\left(\frac{n\sqrt{k}}{n-1}\right)=\mathcal{O}(\sqrt{k}).&\qedhere
    \end{align*}
\end{proof}

In the global setting we get a much higher upper bound on the PoA. It follows from the simple observation that for each terminal a spanning tree suffices to reach it. Hence all nodes together only buy at most $n-1$ time edges per terminal. As a social optimum needs at least $n-1$ edges, the PoA is upper bounded by $k$.

\begin{restatable}{theorem}{PoAUBGlobal}
    \label{thm:PoA_upper_bound_global}
    $\poa^\gl(n,k)\le k$.
\end{restatable}
\begin{proof}
    Let $H$ be a host graph on $n$ nodes containing $k$ terminals. We argue that any NE $\s$ contains at most $kn$ time edges. For each $t\in T_H$ let $B_t$ be the time edge set of a reachability tree towards $t$ in $G(\s)$, i.e., $B_t$ is a tree containing all nodes from $H$, only time edges from $G(\s)$, and all nodes can reach $t$ in $B_t$. Since $\s$ is a NE, $\s$ is a terminal spanner and therefore those trees exist. Any time edge not in any of the $B_t$ can be removed by its buyer since she already has temporal paths to each terminal. Thus,
    \begin{align*}
        \left|\timeEdges_{G(\s)}\right|=\left|\bigcup_{t\in T_H}B_t\right|\le\sum_{t\in T_H}|B_t|\le\sum_{t\in T_H}n-1\le k(n-1).
    \end{align*}
    
    Trivially lower bounding the size of social optima with $n-1$ yields the result.
\end{proof}

For Greedy Equilibria this upper bound is asymptotically tight. This follows from a construction from \citet{AxiotisF16} that creates graphs with $\Theta(n^2)$ minimum temporal spanners. We extend those graphs to get temporal cliques with minimal temporal spanners of size $\Theta(n^2)$, which by \Cref{thm:GETemporalSpanners} implies GEs of the same size. Using \Cref{cor:NE_terminals_non_terminals} and \Cref{lem:NE_dismountable} again generalizes the bound to arbitrarily large graphs with fixed number of terminals $k$.

\begin{restatable}{theorem}{PoALBGlobal}
    \label{thm:PoA_lower_bound_global_GE}
    $\poa_\GE^\gl(n,k)\in \Omega(k)$.
\end{restatable}
\begin{proof}
    In their paper, \citet{AxiotisF16} give an example for dense simple temporal graphs in which all spanners have $\Omega(n^2)$ edges. Hence, those graphs contain inclusion minimal spanners of that size. We now extend their graphs by adding all missing edges and give them the label 1. As minimal temporal spanners stay minimal temporal spanners when adding more edges to the host graph, the resulting graphs still have minimal temporal spanners with $\Omega(n^2)$ edges. By \Cref{thm:GETemporalSpanners}, those spanners imply Greedy Equilibria of the same size. As the added edges also contain a spanning tree, the social optimum on these graphs has $n-1$ edges. Hence, by setting the number of terminals $k$ to $n$, we get an example graph with $\Omega(k)$ times more edges in an equilibrium than in the social optimum. Using \Cref{cor:NE_terminals_non_terminals} and \Cref{lem:NE_dismountable}, we obtain such examples for arbitrarily large number of nodes $n'$.
\end{proof}

\begin{figure*}
    \centering
    \pgfplotsset{colormap={CM}{color=(blue) color=(purple) color=(red) color=(orange)}}
    \tikzset{farbe/.style={color of colormap=#1}}
    \pgfmathsetmacro{\x}{4}
    \pgfmathsetmacro{\xlim}{int(\x-1)}
    \pgfmathsetmacro{\c}{int(2*\x)}
    \pgfmathsetmacro{\clim}{int(\c-1)}
    \pgfmathsetmacro{\halfx}{int(\x/2)}
    \pgfmathsetmacro{\halfxlim}{int(\halfx-1)}
    \pgfmathsetmacro{\halfxlimMinusOne}{int(\halfxlim-1)}
    \resizebox{0.73\linewidth}{!}{
    \begin{tikzpicture}[inner sep=0.5pt]
    	\foreach \i in {0,1,...,\clim}{
    		\pgfmathsetmacro{\deg}{-360*\i/\c+90}
    		\foreach \j in {0,...,\halfxlim}{
    			\pgfmathsetmacro{\offset}{2.5*(\halfx-\j+0.5)}
    			\pgfmathsetmacro{\offsete}{2.5*(\halfx-\j)}
    			\node[draw, circle,minimum size=22pt] (v\i\j) at (\deg:\offset) {$v_{\i,\j}$};
    			\node[draw,circle,minimum size=22pt] (u\i\j) at (\deg:\offsete) {$v'_{\i,\j}$};
    		}
    	}
    	
        \pgfmathsetmacro{\xinc}{int(\x+1)}
    	\foreach \i in {0,1,...,\clim}{
    		\pgfmathsetmacro{\nexti}{int(Mod(\i+1,\c))}
    		\foreach \j in {0,...,\halfxlim}{
    			\foreach \k in {0,...,\halfxlim} {
					\pgfmathsetmacro{\labelv}{int(Mod(2*(\k-\j),\x)+1)}
					\pgfmathsetmacro{\labelu}{int(Mod(2*(\k-\j)+1,\x)+1)}
    				\pgfmathsetmacro{\colorv}{(\labelv-1)*1000/(\x)}
    				\pgfmathsetmacro{\coloru}{(\labelu-1)*1000/(\x)}
    				\path[-, farbe=\colorv]
    				\edc(v\i\j, u\nexti\k, \labelv);
    				\path[-, farbe=\coloru]
    				\edc(u\i\j, v\nexti\k, \labelu);
    			}
                \path[-, farbe=1000]
                \edc(v\i\j, u\i\j, \xinc);
    		}
            \ifthenelse{\x>2}{
            \foreach \j in {0,...,\halfxlimMinusOne} {
                \pgfmathsetmacro{\nextj}{int(\j+1)}
                \path[-, farbe=1000]
                \edc(u\i\j, v\i\nextj, \xinc);
            }}{}
    	}
    \end{tikzpicture}}
    \caption{This figure shows a dense cycle graph for $x=4$.}\label{fig:dense_cycle}
\end{figure*}
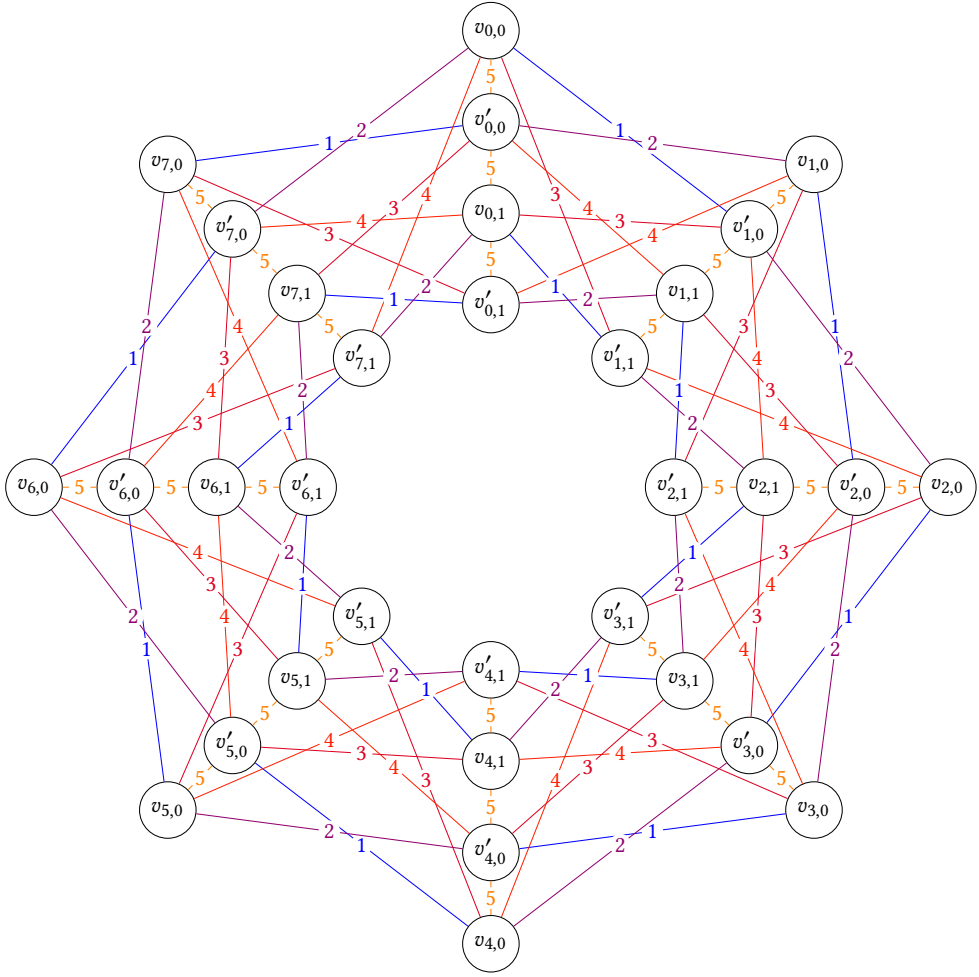

We conclude with one of our main results, establishing a rather high lower bound on the Price of Anarchy of Nash Equilibria in the global setting.
\begin{restatable}{theorem}{PoALB}\label{thm:PoALB}
    $\poa_\NE^\gl(n,k)\in\Omega(\sqrt{k})$.
\end{restatable}
To show this, we construct a class of simple temporal graphs containing $k$ terminals and nodes and $\Omega(k^\frac{3}{2})$ edges, which we term \emph{dense cycle graphs} (see the following \Cref{def:PoALBConstruction} and \Cref{fig:dense_cycle}). We then embed it into a host graph and show that the edges can be assigned to the agents such that the resulting strategy profile is an equilibrium in the global setting. The main idea is that there are $\Theta(\sqrt{k})$ bags arranged in cyclic fashion each containing $\Theta(\sqrt{k})$ nodes. The labeling is such that each node has exactly one temporal path towards the bag on the opposite side enforcing it to buy the whole path of length $\Theta(\sqrt{k})$.

To prove this we prove a sequence of lemmas making up the remaining part of this section. We start with the definition of dense cycle graphs.

\begin{definition}\label{def:PoALBConstruction}
    Let $x\in\N$ be an even number. A \emph{dense cycle graph} is a simple temporal graph $G$ consisting of $2x$ bags $B_0,\dots,B_{2x-1},B_{2x}=B_0$ each containing $\frac{x}{2}$ pairs of nodes, i.e.
    \begin{align*}
        \forall 0\le i < 2x\colon B_i&\coloneqq\{v_{i,0},v'_{i,0}, v_{i,1},v'_{i,1}, \dots, v_{i,\frac{x}{2}-1}, v'_{i,\frac{x}{2}-1}\}\\
        V_G&\coloneqq \bigcup_{0\le i < 2x}B_i.
    \end{align*}
    We call all nodes $v_{i,j}$ \emph{odd} and all nodes $v'_{i,j}$ \emph{even}. The bags form a ring-shaped graph with $n=2x^2$ nodes, where every two adjacent bags are densely connected in the following way:
    \begin{align*}
        E_G\coloneqq{}&\{\{v_{i,j},v'_{i+1,k}\}\mid 0\le i< 2x\wedge 0\le j,k< \frac{x}{2}\}\\
        &\cup\{\{v'_{i,j},v_{i+1,k}\}\mid 0\le i< 2x\wedge 0\le j,k< \frac{x}{2}\}.
    \end{align*}
    Notice, that between two bags, only nodes of different parity are connected. For the labels, we choose $\lambda_G$ such that, for all $0\le i<2x$ and $0\le j,k< \frac{x}{2}$, we have
    \begin{align*}
        \lambda_G(\{v_{i,j},v'_{i+1,k}\})&\coloneqq (2(k-j)\bmod x)+1 \text{ and}\\
        \lambda_G(\{v'_{i,j},v_{i+1,k}\})&\coloneqq ((2(k-j)+1)\bmod x)+1.
    \end{align*}

    Furthermore, we define a \emph{connected} dense cycle graph $G'$ as a cycle graph which contains an additional path with label $x+1$ inside each bag:
    
    \begin{align*}
        E_{G'}=E_G&\cup \{\{v_{i,j},v'_{i,j}\}\mid 0\le i < 2x\wedge 0\le j< \frac{x}{2}\}\\
        &\cup \{\{v'_{i,j},v_{i,j+1}\}\mid 0\le i < 2x\wedge 0\le j< \frac{x}{2}-1\}
    \end{align*}
\end{definition}
For an example of a connected dense cycle graph, see \Cref{fig:dense_cycle}.

To show \Cref{thm:PoALB}, we construct a simple host graph and a NE that contains $\Omega(n^\frac{3}{2})$ edges. We start by defining a specific graph class and showing some useful properties that will eventually yield the desired result.

Towards lower bounding the Price of Anarchy in the global setting, we proof a sequence of useful properties of dense cycle graphs.

\begin{restatable}{lemma}{PoALBSize}\label{lem:PoALBSize}
    Let $x\in \N$ be an even number. Then, a dense cycle graph $G$ for $x$ has $2x^2$ nodes and $x^3$ edges. A connected dense cycle graph $G'$ has $x^3+2x(x-1)$ edges.
\end{restatable}
\begin{proof}
    $G$ consists of $2x$ bags containing $\frac{x}{2}$ node pairs each. Therefore $|V_G|=2x^2$. Between two adjacent bags, there are $x\cdot \frac{x}{2}$ edges, since each node is connected to half of the nodes in the other bag. Therefore $G$ contains $x^3$ edges. $G'$ contains an additional $2x$ paths (one per bag) with $x-1$ edges each.
\end{proof}

\begin{restatable}{lemma}{PoALBIncidentEdges}\label{lem:PoALBIncidentEdges}
    Let $G$ be a dense cycle graph for some even $x\in\N$. Then, for each node $v\in V$ and label $1\le l\le x$, there is exactly one edge with label $l$ incident to $v$. Moreover,
    \begin{compactitem}
        \item for odd nodes $v_{i,j}$, all incident odd-labeled edges go towards bag $B_{i+1}$ and all incident even-labeled edges go towards bag $B_{i-1}$,
        \item for even nodes $v'_{i,j}$, all incident even-labeled edges go towards bag $B_{i+1}$ and all incident odd-labeled edges go towards bag $B_{i-1}$.
    \end{compactitem}
\end{restatable}
\begin{proof}
    In \Cref{def:PoALBConstruction}, we see that in dense cycle graphs, even nodes are only adjacent to odd nodes and vice versa. Let $v_{i,j}$ be an odd node. There are $\frac{x}{2}$ edges going towards bag $B_{i+1}$ with labels
    \[\{(2(k-j)\bmod x)+1\mid 0\le k< x\}=\{1,3,\dots,x-1\}.\]
    
    Similarly, there are $\frac{x}{2}$ edges going towards bag $B_{i-1}$ with labels
    \[\{((2(j-k)+1)\bmod x)+1\mid 0\le k<x\}=\{2,4,\dots,x\}.\]

    The proof for even nodes follows the same arguments.
\end{proof}

\begin{restatable}{lemma}{PoALBPaths}\label{lem:PoALBPaths}
    Let $G$ be a dense cycle graph for some even $x\in\N$. Then, each odd node $v_{i,j}$ has exactly one temporal path $P_{i,j}$ into the opposite bag $B_{(i+x)\bmod 2x}$ and exactly one temporal path $Q_{i,j}$ into bag $B_{(i+x+1)\bmod 2x}$. Moreover,
    \begin{compactitem}
        \item $P_{i,j}$ goes through bags $B_i, B_{i+1}, \dots, B_{(i+x)\bmod 2x}$,
        \item $Q_{i,j}$ goes through bags $B_i, B_{i-1}, \dots, B_{(i+x+1)\bmod 2x}$, and
        \item the edges of all paths $P_{i,j}$ form a disjoint union of $E_G$.
    \end{compactitem}
    Due to symmetry, similar properties hold for all even nodes $v'_{i,j}$.
\end{restatable}
\begin{proof}
    A path $P_{i,j}$ that goes from a node $v_{i,j}$ through bags $B_i, B_{i+1}, \dots, B_{(i+x)\bmod 2x}$ has at least $x$ edges. From \Cref{lem:PoALBIncidentEdges} we know that the labels of edges incident to a node are unique. Therefore, $P_{i,j}$ must contain the labels $1$ to $x$. Again with the help of \Cref{lem:PoALBIncidentEdges}, such a path can be constructed easily: Starting at $v_{i,j}$ we take the edge with label 1 from $v_{i,j}$ to bag $B_{i+1}$, next the edge with label 2 to $B_{i+2}$ and so on. Since choosing the next edge is unique in each step, the whole path is unique.

    The existence and uniqueness of the path $Q_{i,j}$ going from $v_{i,j}$ through bags $B_i, B_{i-1}, \dots, B_{(i+x+1)\bmod 2x}$ follows the same arguments with the key difference of the first having label 2 instead of 1 since $v_{i,j}$ is an odd node and therefore only has edges with even labels towards $B_{i-1}$.

    It remains to show, that the edges of all $P_{i,j}$ form a disjoint union of $E_G$. Let $e\in E_G$ be an arbitrary edge between bags $B_i$ and $B_{i+1}$. Similar to the construction of the paths above, we can walk "backwards": Starting at the endpoint of $e$ in $B_i$, take the edge with label $\lambda(e)-1$ between bags $B_{i-2}$ and $B_{i-1}$, then the edge with label $\lambda(e)-2$ and so on. Because of \Cref{lem:PoALBIncidentEdges}, this process is unique and will stop at some node $v$ after traversing an edge with label 1. Because of the direction of our process, $v$ has to be an odd node. This means, that $e$ is part of one of those paths. Due to the uniqueness of the process, $e$ cannot be part of another path, concluding the proof.
\end{proof}

\begin{restatable}{lemma}{PoALBConnected}\label{lem:PoALBConnected}
    Let $G$ be a connected dense cycle graph for some even $x\in\N$. Then, $G$ is temporally connected.
\end{restatable}
\begin{proof}
    Let $v\in V$ be a node. Because of \Cref{lem:PoALBPaths}, there are temporal paths from $v$ into each bag with an arrival time of at most $x$. Since there is a path with label $x+1$ in each bag, $v$ can reach every node in every bag.
\end{proof}

\begin{restatable}{lemma}{PoALBNash}\label{lem:PoALBNash}
    For each sufficiently large $n'\in \N$, there is a simple host graph $H$ with $2n' \ge n\ge n'$ nodes and $k=n$ terminals containing a Nash Equilibrium $\s$ s.t. $G(\s)$ is a connected dense cycle graph.
\end{restatable}
\begin{proof}
    Let $n'\in\N$ and $G$ be a connected dense cycle graph for $x=2\lceil\sqrt{\frac{n'}{8}}\rceil$. Note that $G$ has $2n' \ge 2x^2 \ge n'$ nodes. We define the host graph $H$ as the complete simple temporal graph on $V_G$ with
    \begin{align*}
        \forall e\in E_H\colon \lambda_H(e)=\begin{cases}
            \lambda_G(e)&\text{, if } e\in E_G\\
            x+2 &\text{, else}.
        \end{cases}
    \end{align*}
    Additionally, we construct the strategy profile $\s$ in the following way. For each odd node $v_{i,j}$, we set $S_{v_{i,j}}=\timeEdges_{P_{i,j}}$ where the $P_{i,j}$ are the paths from \Cref{lem:PoALBPaths}. For the first even node $v'_{i,0}$ in each bag, we set $S_{v'_{i,0}}$ to the set of (time) edges with label $x+1$ in the bag $i+x\bmod 2x$; this is the path with label $x+1$ in the bag on the opposite side of the connected dense cycle graph. For all of the other even nodes, the strategy is empty. By \Cref{lem:PoALBPaths}, every edge is bought by exactly one node with these strategies. Also, it holds $G(s)= G$. By \Cref{lem:PoALBConnected}, $G(\s)$ is temporally connected. It remains to show that no node has an improving response.

    Consider an odd node $v_{i,j}$. This node has to reach all of the nodes in the opposite bag $B_{(i+x) \bmod 2x}$. By $\Cref{lem:PoALBPaths}$, $P_{i,j}$ is the only temporal path in $G$ to that bag. All other paths to that bag in $H$ must therefore include an edge with label $x+2$. As $S_{v_{i,j}}=E_{P_{i,j}}$, an improving response cannot contain $E_{P_{i,j}}$ and therefore each path to a node in $B_{(i+x) \bmod 2x}$ has to contain an edge of label $x+2$. As $s$ does not contain any edges with label at least $x+2$, $v_{i,j}$ has to include those edges into its strategy. It needs one label $x+2$ edge for each node in $B_{(i+x) \bmod 2x}$. As there are $x$ nodes in $B_{(i+x) \bmod 2x}$, there cannot be an improving response.

    Now consider an even node $v_{i,j}$. For all $j \neq 0$ these nodes do not buy edges, so they cannot have an improving response. For $j=0$, the node $v_{i,j}$ has to reach all nodes in the opposite bag $B_{(i+x) \bmod 2x}$. By $\Cref{lem:PoALBPaths}$, there is only one node in the opposite bag that $v_{i,j}$ can reach before time $x+1$. So for the other nodes it needs one edge of label at least $x+1$ each. As $s_{-v_{i,j}}$ does not contain any edges of label at least $x+1$ adjacent to nodes in $B_{(i+x) \bmod 2x}$, those edges have to be bought by $v_{i,j}$, so it does not have an improving response.
\end{proof}

With this, we can now obtain the bound on the PoA.
\PoALB*
\begin{proof}
    We aim to show that there exists a constant $c \in \R$ such that for all sufficiently large $n,k \in \N$ with $n \geq k$ holds that $\poa^\lo(n,k) \ge c \sqrt{k}$. To do so, we first show that the statement holds for infinitely many pairs $(n,k)$ and then fill the gaps with \Cref{lem:NE_dismountable}.
    
    Let $n,k \in \N$ be sufficiently large. By \Cref{lem:PoALBNash} there exists a simple host graph $H$ with $k'\in \N$, $k \ge k' \ge k/2$ nodes and terminals that contains a Nash Equilibrium $s$ that is a connected dense cycle graph. By \Cref{lem:PoALBSize}, this equilibrium has $m \ge k'\sqrt{k'}/2$ time edges. As $k \le n$, there exists a constant $c' \in \N$ such that $n - k + k' \ge c'k' \ge n/4$. By definition of $H$ and \Cref{cor:NE_terminals_non_terminals}, there exists a host graph $H'$ with $n' = c'k'$ nodes and $k'$ terminals containing $m' = c'm + (c'-1)k'$ time edges. Additionally, if $c' \ge 3$, $H'$ contains a spanning tree. Using \Cref{lem:NE_dismountable} repeatedly on $H'$, we create a graph with $n$ nodes, $k$ terminals and an equilibrium with at least $m'$ time edges that also has a spanning tree. As $m' \ge c'm \ge \frac{n'}{k'}\cdot k'\sqrt{k'}/2 \ge n'\sqrt{k'}/2 \ge n\sqrt{k}/16 $, it follows that $\poa^\lo(n,k) \ge \sqrt{k}/16$. If $c' < 3$ we use the same argument on $H$ directly to get a similar bound.
\end{proof}


\section{Conclusion and Outlook}
We study a game-theoretic model where non-cooperative agents form a temporal network. We introduce two new dimensions to the problem: (i) the agents can build any edge in the network, (ii) the agents want to reach a subset of the other agents. We analyze the existence, the structure and the quality of equilibria. Our main results are upper and lower bounds on the PoA, along with a powerful tool that allows us to translate lower bounds between the non-terminal case to the terminal one. We show that the global and the local model are incomparable, which contradicts the intuition that more power to the agents would improve the quality of the equilibria. Additionally, all of our results hold for both the single label and multi label case.

There are various open problems that stem from our work. With regards to improving our results, it would be interesting to see whether equilibria exist for more than two terminals. We conjecture that this is true, but even the case with three terminals proves to be very challenging. It is also worthwhile to close the gaps between the upper and the lower bounds on the PoA for the local and the global setting. For Greedy Equilibria we showed that the PoA in the global setting is strictly worse than the PoA in the local setting. We believe that this is also true for Nash Equilibria. In particular, we believe that the PoA for NE is close to the current lower bound of $\Omega(\log(k))$ in the local setting while we conjecture it to be close to the current lower bound of $\Omega(\sqrt{k})$ in the global setting. This conjecture also implies that the gap between the PoA for NE and the PoA for GE is much larger in the global setting ($\Omega(\sqrt{k})$ versus $\Theta(k)$) than in the local setting (where the PoA for GE is at most a $\log$ factor larger).

While our paper provides a very general model for studying the formation of temporal networks by non-cooperative agents, there are still more extensions to be investigated. For example, agents might want to minimize the distance to the other agents or share the costs of buying edges.
Structural properties of the host network could also be considered. For enhanced applicability, the edges could be directed, could have non-uniform buying costs, and/or non-instant traversal times.

\bibliographystyle{ACM-Reference-Format} 
\bibliography{aamas_2025}

\end{document}